\documentclass[aps,pra,graphicx,10pt]{revtex4-2}
\usepackage{amssymb,amsbsy,times,fancyhdr,color}
\usepackage{amsmath,amsthm}
\usepackage{mathrsfs,amsfonts}
\usepackage{latexsym,afterpage}
\usepackage{graphicx,subfigure,multirow}

\newtheorem{prop}{Proposition}

\usepackage{lineno}

\begin{document}


\title{Transient growth and Counter-propagating Rossby Waves in a sloped Eady problem} 


\author{J. Yin}
\email[]{jyinar@connect.ust.hk}
\affiliation{Department of Ocean Science, Hong Kong University of Science and Technology, Hong Kong}


\author{N. Harnik}


\author{E. Heifetz}
\affiliation{Porter school of the Environment and Earth Sciences, Tel Aviv University,
 69978, Israel}
 

\author{J. Mak}
\email[]{julian.c.l.mak@googlemail.com}
\affiliation{Department of Ocean Science, Hong Kong University of Science and Technology, Hong Kong}
\affiliation{Department of Mathematics, Hong Kong University of Science and Technology, Hong Kong}
\affiliation{National Oceanography Centre, Southampton, UK}




\begin{abstract}
A mechanistic interpretation for baroclinic normal mode instability with exponential growth is the interaction arising from Counter-propagating Rossby Waves (CRWs), where pairs of CRWs phase-lock and interfere constructively. On the other hand, an arbitrary initial condition evolving towards a phase-locked state, or even when the system does not allow phase-locking of CRWs, can lead to finite transient growth as long as there is constructive interference, which over a time window of interest could in principle exceed the gain associated with exponential growth for some choice of norm. The present work investigates the characteristics of finite time transient instabilities and its description in terms of CRWs in a baroclinic system with asymmetric interaction, in contrast to usual settings where there is an inherent symmetry that leads to equal and opposite CRWs. A sloped Eady problem on a $f$-plane is considered here, which has the benefit that there are only two CRWs in the system, and the presence of a bottom topographic slope can be formulated in a such a way that only one of the CRW is affected and controlled by a single parameter representing the asymmetry, via modifying the bottom background potential vorticity gradient. We derive conditions where the instantaneous growths have key CRW properties that are independent of the asymmetry parameter for a reasonably general family of norms, and that the instability characteristic in the total energy norm is entirely independent of the asymmetry parameter. For finite time transient growth, the optimal initial condition is in some sense counter to that of the exponential growing normal mode: where the bottom slope is such that it is strengthening the bottom CRW, the normal mode has a larger bottom CRW, but the finite time optimal initial condition wants the bottom CRW to be weak instead, and vice-versa. While the results presented here are for the sloped Eady problem, the same results were also found for an analogous Phillips problem, and we speculate the instability characteristics may have origins in the underlying symmetries that are common to those two problems.
\end{abstract}


\maketitle 



\section{Introduction}

Many fluid systems in geophysical and astrophysical settings are influenced by rotation and buoyancy effects, and the dynamics and turbulence in these baroclinic systems are of interest in how it shapes stellar and planetary global circulation \cite[e.g.,][]{Vallis-GFD, Hughes-et-al-Tachocline, Read-et-al20}, for example. In the context of the Earth system, baroclinic turbulence impacts the dynamical features we observe in the atmosphere and the ocean, with implications particularly for climate modeling. The role of topography provides a strong constraint and shapes the currents and dynamics in the ocean, and there have been numerous works looking into understanding the mechanisms, conditions and characteristics for baroclinic instability \cite[e.g.,][]{BlumsackGierasch72, Mechoso80, Isachsen11, Brink12, ChenKamenkovich13, Isachsen15, Pedlosky16, Chen-et-al20, Tanaka21, Mak-et-al24, Lobo-et-al25}, its transition to turbulence and eventual saturation, and the associated eddy/wave-mean flow interaction in relation to parameterization \cite[e.g.,][]{ManucharyanIsachsen19, Chen-et-al20, Tanaka21, WeiWang21, Wei-et-al22, Wei-et-al24, DengWang24, NummelinIsachsen24, Sterl-et-al24, Sterl-et-al25, Lobo-et-al25, PudigSmith25, DengWang26}. For the present work we will focus on the instability aspect.

The three canonical set ups for the baroclinic instability problem are that due to Charney \cite{Charney47}, Eady \cite{Eady49} and Phillips \cite{Phillips54}, under the quasi-geostrophic approximation described by material advection of potential vorticity (PV) \cite[e.g.,][]{Vallis-GFD}. While the three problems illustrate different characteristics of baroclinic instability, one aspect that unifies them is that, mechanistically, they can interpreted to arise from the constructive interaction of Counter-propagating Rossby Waves (CRWs) \cite[e.g.,][]{Bretherton66a, Hoskins-et-al85, Heifetz-et-al04a}. Briefly, Rossby waves are supported on PV gradients associated with the basic state, and Rossby waves themselves have self-induced PV and so velocity anomalies associated with it that leads to self-advection (cf. individual planetary Rossby waves propagate in a retrograde fashion because of this mechanism). There can be cases where Rossby waves counter-propagate against a background mean flow, as well as a pair of Rossby waves that can mutually interact because of the induced velocity associated with the PV anomalies (action-at-a-distance, cf. electric charge and its electric field). In such cases where there is phase-locking as well as mutual constructive interaction, both waves grow in magnitude, and is a physical mechanism invoked for baroclinic normal mode instability with exponential growth (see Fig.~\ref{fig:fig1} for a schematic of the mechanism but in an asymmetric set up of interest here). The Charney--Stern--Pedlosky theorem \cite{CharneyStern62, Pedlosky64a, Pedlosky64b} is then argued to be a \emph{necessary} (but not sufficient) condition for existence of Rossby wave pairs that can counter-propagate because of the basic state set up; it is only a \emph{necessary} condition, because counter-propagation alone does not imply phase-locking or constructive mutual interference. The instabilities with exponential growths in the Charney, Eady and Phillips problems can then be interpreted as arising from interaction between a boundary and interior CRW, between two boundary CRWs, and between two interior CRWs respectively. We note that the latter two are usually considered in modern textbooks, partly due to them having simple closed form solutions: briefly, the Eady and the Phillips problem have PV gradients set up such that there are only two CRWs in the system (because of the $f$-plane set up in the Eady problem, and the two-layer quasi-geostrophic equation set up in the Phillips problem), while the Charney problem, formulated on a $\beta$-plane, have PV perturbations and thus CRWs throughout the domain (although the one at the critical layer is the most prominent, \cite[e.g.,][]{Heifetz-et-al04b}). Note that the CRW mechanism and associated stability conditions where applicable were demonstrated to operate in other shear flow instabilities (e.g., incompressible shear flow \cite[e.g.,][]{Heifetz-et-al99, Heifetz-et-al09, BiancofioreGallaire12}, stratified shear flows \cite[e.g.,][]{Harnik-et-al08, Rabinovich-et-al11, Carpenter-et-al11}, magnetohydrodynamic flows \cite[e.g.,][]{Heifetz-et-al15}), with differences in the wave properties and interaction mechanism depending on the physics of the system.

For baroclinic systems, the presence of a bottom slope could be accounted for in a simple fashion by making an assumption that the slope gradient is small (more precisely, of order Rossby number). In that case the effect of the slope makes an appearance as an additional term in the background PV gradient, and assumption of linear meridional slope results in a single parameter controlling the strength that either adds to or subtracts from the bottom background PV gradient. This approach was taken for a number of works on linear normal mode instability for the Eady problem \cite[e.g.,][]{BlumsackGierasch72, Isachsen11, Pedlosky16, Chen-et-al20, Mak-et-al24}. It was alluded to in the works of \cite{Isachsen11, Chen-et-al20} that the presence of meridional slopes would affect the CRWs and thus impact the normal mode instability characteristics, but \cite{Mak-et-al24} clarifies the mechanism in a way that is consistent over parameter space for the sloped Eady problem. For the setting where the slope \emph{strengthens} the bottom PV gradient, the bottom CRW becomes stronger, propagates faster \emph{as well as} leading to a stronger interaction; this case is roughly illustrated in Fig.~\ref{fig:fig1} of the present article. The associated strengthening can be offset by going to larger wavenumbers since shorter waves propagate slower and have a weaker interaction, at the expense of a reduction in the exponential growth rate; the offsetting can always occur, albeit over an increasingly small bandwidth in wavenumber and with decreasing exponential growth rates. The opposite is largely true for the case where the slope \emph{strengthens} the bottom PV gradient, except that while a weakened bottom CRW can be compensated by going to smaller wavenumbers, longer waves also propagate faster and can overwhelm the background flow contribution to counter-propagation; beyond a certain point, there is simply no counter-propagation possible at any wavenumber and so no exponential growth is possibly (cf. when Charney--Stern--Pedlosky no longer holds). The reader is referred to Fig.~8 of \cite{Mak-et-al24} for a schematic of the resulting CRW interaction, and the main observation there is that there is a controlling parameter leads to an \emph{asymmetry}, where the strengthening case behaves differently to the weakening case.

While normal mode instabilities are the dominant features in the asymptotic long time regime, the assumption of exponential dependence in time aids in making analytical progress for the shear instability problem but is at heart a simplifying assumption. Physically, perturbations may have \emph{transient growth} or \emph{response} in some choice of measure/norm that could be larger for \emph{finite} time windows. Mathematically, the exponential solutions do not form a complete set when the governing operator is \emph{non-normal}, and the initial-value problem is not completely described by the normal modes. In the geophysical fluid dynamics literature at least, this was first pointed out in the works of \cite{Farrell82, Farrell84} for the standard Eady problem. The associated theory is well-developed \cite[e.g.,][]{Trefethen-et-al93, FarrellIoannou96a, FarrellIoannou96b} and has been considered by many researchers particularly in the field of fluid dynamics and in the study of transition to turbulence; we refer the reader to the reviews of \cite{Schmid07, Kerswell18} for an overview of the theory and applications.

A motivating interest in our work is on understanding how the slope affects characteristics of transient growth in the sloped Eady problem, given the asymmetric impact it has on normal mode instabilities \cite{Mak-et-al24}, and how such transient growth would be interpreted in terms of CRWs. Transient growth has been considered in many works for the non-sloped Eady problem \cite[e.g.,][]{Farrell82, Farrell84, DaviesBishop94, DeVriesOopsteegh07a, DeVriesOopsteegh07b, DeVries-et-al09, KalashnikChkhetiani18}, and an interpretation in terms of CRWs is the growth ones gets as the CRWs start from a non-phased-locked configuration and evolve towards a phase-locked position in the case where there is normal mode instability, or the growth one gets when CRWs propagate and interact with each other without phase-locking in settings where there is no normal mode instability \cite[e.g.,][]{HeifetzMethven05, KalashnikChkhetiani18}. In terms of a dynamical system representation, the goal of transient instability is finding the initial location in phase space such that the resulting trajectory has maximum growth in some norm to be specified as it evolves in phase space under the dynamical evolution. The previous works have considered cases where the system is chosen to be symmetric at the system level (usually in the base PV gradient), where it is the asymmetry in the choice of initial condition that is considered, calculating the resulting growth and/or patterns via the initial value problem \cite[e.g.,][]{DaviesBishop94, DeVriesOopsteegh07a, DeVriesOopsteegh07b, DeVries-et-al09, KalashnikChkhetiani18}. Here our focus is investigating the role of the asymmetry in the base PV gradient, controlled in this case by the magnitude of the bottom slope, and how that affects the resulting \emph{optimal} transient growth and its characteristics. 

For our present work we will present results for an investigation of transient baroclinic instabilities for the sloped Eady problem, following the same set up of \cite{Mak-et-al24}; a similar investigation has been performed for an analogous Phillips problem \cite[cf.,][]{ChenKamenkovich13} with and without the planetary $\beta$, but the qualitative findings are exponentially identical to the sloped Eady problem, and we refer the reader to the discussion section Appendix~\ref{app:A} and the provided analysis code for related information. In Sec.~\ref{sec:eady} we state and provide an overview of the governing equations from \cite{Mak-et-al24} for the sloped Eady problem, in the forms relevant for the study of transient growth and the interpretation in terms of CRWs. In Sec.~\ref{sec:transient} we recap the on the mathematical details of studying transient growth, and provide some results relating to norm independence of instability characteristics associated with maximum instantaneous growth, a seemingly unique role of the total energy norm, and asymmetry effects introduced by the presence of the small linear slope on finite time transient growth in terms of CRW interactions. We conclude and provide outlooks in Sec.~\ref{sec:conclude}.


\section{Governing equations for the sloped Eady problem}\label{sec:eady}

The setup and relevant equations are essentially given in \cite{Mak-et-al24}, but we reproduce some of the details for internal consistency. We assume the dynamics are rotationally dominant and in the small Rossby number regime such that it is governed by the continuously stratified quasi-geostrophic equations \cite[e.g.,][]{Vallis-GFD}, with PV advection in the interior of the domain and buoyancy advection at the vertical boundaries. The physical set up is as given in the schematic in Fig.~\ref{fig:fig1}, noting that the current set up considers a symmetrized domain (as opposed to having the bottom boundary at $z=0$ as in the atmospheric setting, or upper boundary at $z=0$ for oceanic settings). We take the $f$-plane formulation with the Coriolis parameter $f_0 > 0$, and the basic state shear flow $\overline{U}(z)\boldsymbol{e}_x$ in the zonal direction with a vertical linear shear, i.e. $\overline{U}(z) = \Lambda z$; by geostrophic balance, the basic meridional buoyancy profile is $\overline{B}(y) = -f_0 \Lambda y$. There is a linear meridional bottom slope, which we assume to be of order Rossby number \cite[e.g.,][]{Isachsen11, Chen-et-al20}, and its impact manifests as a contribution in the boundary condition via modifying the bottom background PV gradient.

\begin{figure}
  \includegraphics[width=\textwidth]{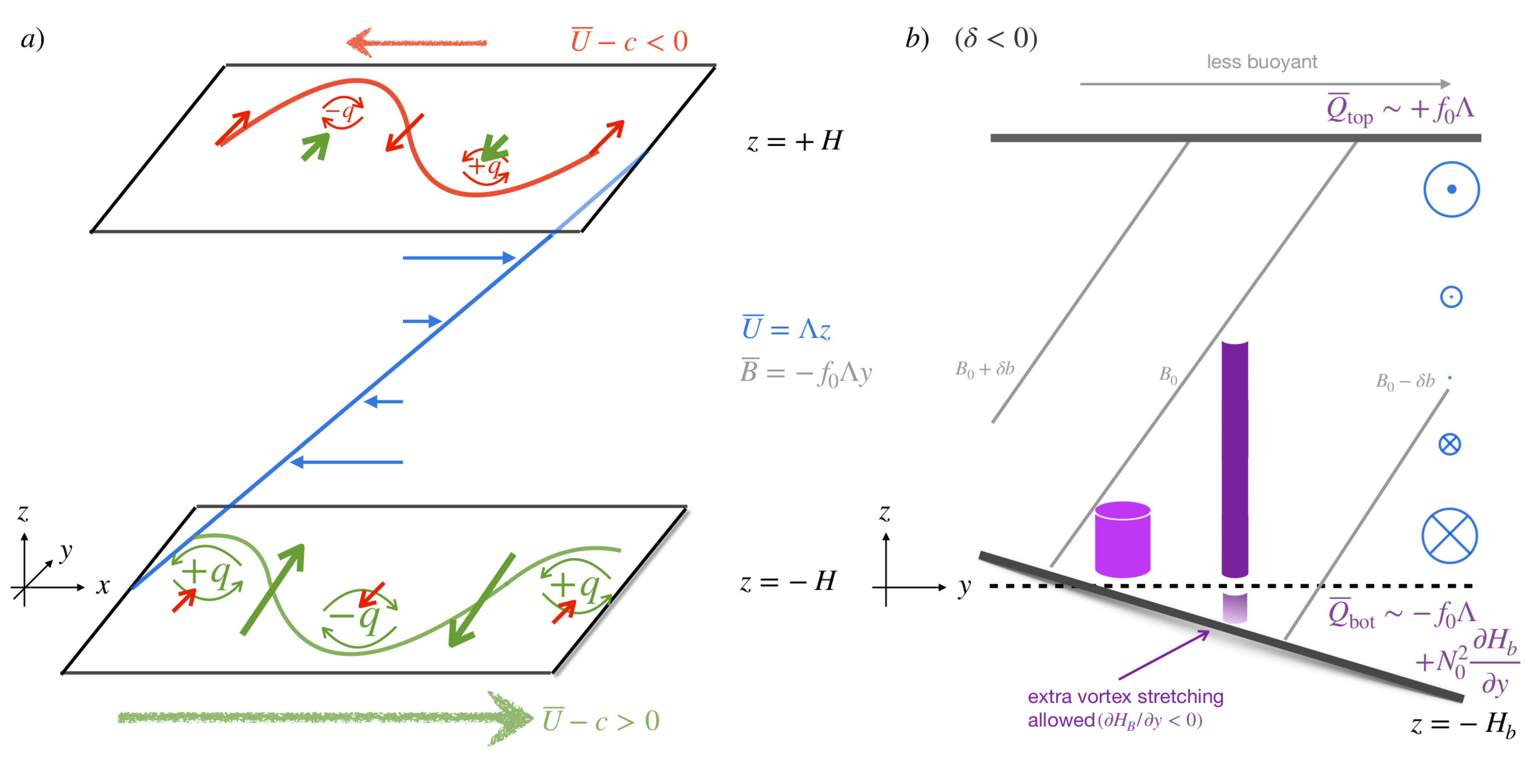}
  \caption{Physical schematic for the sloped Eady problem for the $\delta < 0$ case (where topography and $\overline{B}$ are sloping in the opposite sense) that \emph{adds} to the bottom PV gradient, reworked from Fig.~1 of \cite{Mak-et-al24}. ($a$) Predominantly $(x,z)$ cross-section view showing basic state and CRWs localized at the vertical boundaries with displacement in the meridional direction $y$, in an unstable configuration for a case where the bottom CRW is \emph{strengthened} (indicated by the relatively larger self-induced velocities with longer and thicker arrows). ($b$) Head on $(y,z)$ view, illustrating the basic linear flow $\overline{U}$ coming out and going into the page, the associated buoyancy profile via thermal wind shear relation, where $\delta < 0$ case is such that the slope \emph{reinforces} the bottom PV gradient by allowing extra stretching of a fluid column relative to the flat case. The $\delta > 0$ case would have $\overline{B}$ and slopes sloping in the same sense, leading to a \emph{weaker} bottom PV gradient and weaker bottom CRW.}
  \label{fig:fig1}
\end{figure}

Linearizing about the basic states and non-dimensionalizing by horizontal length scale $L$, vertical length scale $H$ and time scale $T = L / (\Lambda H)$, we obtain the governing equations in terms of a streamfunction $\psi$ as (Eq.~3 of \cite{Mak-et-al24})
\begin{linenomath*}
\begin{subequations}\label{eq:lin-nondim}
\begin{align}
  &\left(\frac{\partial}{\partial t} + z\frac{\partial}{\partial x}\right)\left(\nabla^2 \psi + F^2 \frac{\partial^2\psi}{\partial z^2}\right)=0, \qquad &z \in(-1, 1),\\
  &\left(\frac{\partial}{\partial t} + \frac{\partial}{\partial x}\right)\frac{\partial\psi}{\partial z} - \frac{\partial\psi}{\partial x} = 0, \qquad &z = 1,\\
  &\left(\frac{\partial}{\partial t} - \frac{\partial}{\partial x}\right)\frac{\partial\psi}{\partial z} - (1 - \delta)\frac{\partial\psi}{\partial x} = 0, \qquad &z = -1,
\end{align}
\end{subequations}
\end{linenomath*}
with $F^2 = (f_0 L / NH)^2$ the inverse of the Burger number, which is a scaling factor and we take it to be unity for simplicity. Here the interior PV is zero and the interior equation dictates the vertical structure of the solutions (which are hyperbolic sine and cosines in $z$), while the boundary conditions provide the remaining information. The key control parameter is
\begin{linenomath*}\begin{equation}
  \delta = \frac{\partial H_b}{\partial y}\left(\frac{N_0^2}{f_0 \Lambda}\right) = \left.\frac{\partial H_b}{\partial y} \right/ \frac{-\partial \overline{B} / \partial y}{\partial \overline{B} / \partial z},
\end{equation}\end{linenomath*}
which relates to the orientation between the slopes of the background buoyancy (the denominator including the minus sign) and the linear bottom topography (the nominator); when $\delta < 0$ (or retrograde configuration) the orientations are opposite, while $\delta > 0$ (or prograde configuration) is the converse. The $\delta$ parameter can be related to the topographic $\beta$ \cite[e.g.,][]{Isachsen11}, which is reasonable given the clear role of $\delta$ in adding or subtracting from the background PV gradient, as illustrated in Fig.~\ref{fig:fig1}. 

We assume periodicity in $x$ and consider solutions of the form $\psi = \tilde{\psi}(z, t)\exp[\mathrm{i}kx]g(y)$, where $\mathrm{i} = \sqrt{-1}$ and $g(y)$ is an appropriate eigenfunction of the Laplacian operator so $\partial^2 g / \partial y^2 = -l^2 g$ (e.g., combinations of $\sin ly$ and $\cos ly$ as appropriate with periodic or no-normal flow meridional boundary conditions), so that $(k,l)$ are the zonal and meridional wavenumbers. One would normally proceed by staying with streamfunction variables, taking a wave-like dependence in time, and then compute for the normal mode instabilities with exponential growth. However as argued in \cite{Mak-et-al24} one could reformulate the problem in terms of PV variables, which is equivalent and more natural for interpretation within the CRW framework, particularly for cases where we know the CRWs are localized and finite in number. To do this, we consider a linear superposition of untilted structures localized at the top and bottom boundary (subscript $T$ and $B$ respectively) $\tilde{\psi} = \tilde{\psi}_T + \tilde{\psi}_B$. By PV inversion (cf. Eq.~\ref{eq:lin-nondim}$a$), we have
\begin{linenomath*}\begin{equation}\label{eq:helmholtz}
  \tilde{q} = -\mu^2 \tilde{\psi} + \frac{\partial^2\tilde{\psi}}{\partial z^2}, \qquad \mu^2 = k^2 + l^2,
\end{equation}\end{linenomath*}
and there is no planetary $\beta$ term in this case since our formulation is on a $f$-plane. Since PV inversion is linear, formally we have $\tilde{q} = \tilde{q}_T + \tilde{q}_B$, with the inversion subject to the boundary conditions
\begin{linenomath*}\begin{equation}\label{eq:bcs}
  \left.\frac{\partial \tilde{\psi}_T}{\partial z}\right|_{z=-1} = 0, \qquad \left.\frac{\partial \tilde{\psi}_B}{\partial z}\right|_{z=+1} = 0.
\end{equation}\end{linenomath*}
Taking the $\tilde{q}_{T,B} = \hat{q}_{T,B}(t)\hat{\delta}(z\mp 1)$ so that the PV signature is localized at the top and bottom boundary $z=\pm 1$ respectively, we have
\begin{linenomath*}\begin{equation}\label{eq:psi_wave}
  \tilde{\psi}_T = -\hat{q}_T\frac{\cosh\mu(1+z)}{\mu\sinh2\mu}, \qquad \tilde{\psi}_B = -\hat{q}_B\frac{\cosh\mu(1-z)}{\mu\sinh2\mu},
\end{equation}\end{linenomath*}
either by manual construction \cite[cf.][]{DaviesBishop94}, or by noting that we are looking for the Green's function associated with the one-dimensional Helmholtz operator subject to appropriate homogeneous Neumann conditions. With the above identity and algebraic manipulations, the boundary equations in \eqref{eq:lin-nondim} can be written as the system
\begin{linenomath*}\begin{equation}\label{eq:A_equ}
  \frac{\mathrm{d}}{\mathrm{d}t}\mathbf{q} = \mbox{A}\mathbf{q}, \qquad \mathbf{q} = \begin{pmatrix}\hat{q}_T \\ \hat{q}_B\end{pmatrix}, \quad \mbox{A} = \mathrm{i}k \begin{pmatrix} -\left(1 - \cfrac{\cosh 2\mu}{\mu \sinh 2\mu}\right) & \cfrac{1}{\mu\sinh 2\mu} \\ -\cfrac{1 - \delta}{\mu \sinh 2\mu} & 1 - (1-\delta)\cfrac{\cosh 2\mu}{\mu\sinh 2\mu} \end{pmatrix}.
\end{equation}\end{linenomath*}
If we consider solutions $\hat{q}_{T,B}(t) \sim \exp[\mathrm{i}\omega t] = \exp[\mathrm{i}k(c_r + \mathrm{i}c_i)t]$ then it can be shown the resulting dispersion relation has solutions equivalent to standard formulations in terms of streamfunction variables (e.g., Eq.~8 and Fig.~2 of \cite{Mak-et-al24}). Note also that the governing equation \eqref{eq:A_equ} is $\mbox{PT}$ (parity-time) symmetric, as previously noted in \cite{Mak-et-al24}: following \cite{David-et-al22}, $\mbox{P}:(x, y) \mapsto (-x, -y)$ leads to $(k, l) \mapsto (-k, -l)$ while $\mbox{T}:(t, q) \mapsto (-t, -q)$. Under $\mbox{PT}$, $\mbox{A} \mapsto -\mbox{A}$ and there is an extra minus sign associated with $\mathrm{d}/\mathrm{d}t$, leaving the system invariant. If we consider modal solutions, then $\mathrm{d}/\mathrm{d}t \to \mathrm{i}\omega$, and then the resulting operator (without the factor of $\mathrm{i}$ that cancel out with modal solutions) can be shown to be $\mbox{PT}$ symmetric, i.e. $(\mbox{PT}) \mbox{A}^* (\mbox{PT})^{-1} = (-\mbox{I}) \mbox{A}^* (-\mbox{A}) = \mbox{A}^* = \mbox{A}$, where $(\cdot)^*$ is taking a component-wise complex conjugate (and is trivial in this case since the resulting $\mbox{A}$ only has real entries).

For formulating the problem in terms of CRW quantities, we take $\hat{q}_T = T\mathrm{e}^{\mathrm{i}\epsilon_{T}}$ and similarly for $\hat{q}_B$, where the amplitudes $T$, $B$ and the phases $\epsilon_{T,B}$ are real functions of time. Taking the real and imaginary parts accordingly leads to (correcting typographic sign errors in Eq.~15 of \cite{Mak-et-al24})
\begin{linenomath*}\begin{subequations}\label{eq:amp_phase}
\begin{align}
  \frac{1}{T}\frac{\mathrm{d} T}{\mathrm{d} t} &= +\frac{k}{\mu \sinh 2\mu}\frac{B}{T}\sin\Delta\epsilon,\\
  \frac{1}{B}\frac{\mathrm{d} B}{\mathrm{d} t} &= +\frac{k(1-\delta)}{\mu \sinh 2\mu}\frac{T}{B}\sin\Delta\epsilon,\\
  -\frac{1}{k}\frac{\mathrm{d} \epsilon_T}{\mathrm{d} t} &= +\left[1 - \frac{1}{\mu\sinh2\mu}\left(\cosh2\mu + \frac{B}{T}\cos\Delta\epsilon\right)\right],\\
  -\frac{1}{k}\frac{\mathrm{d} \epsilon_B}{\mathrm{d} t} &= -\left[1 - \frac{(1-\delta)}{\mu\sinh2\mu}\left(\cosh2\mu + \frac{T}{B}\cos\Delta\epsilon\right)\right].
\end{align}
\end{subequations}\end{linenomath*}
Defining the relative amplitude ratio $\tan\Gamma = T/B$ and relative phase-shift $\Delta\epsilon = \epsilon_T - \epsilon_B$ ($\Delta\epsilon > 0$ means the top wave has a PV signature that \emph{lags} the bottom wave PV; cf. Fig.~\ref{fig:fig1}$a$), it can be shown that
\begin{linenomath*}\begin{subequations}\label{eq:amp_phase_ratios}
\begin{align}
    \frac{\mathrm{d} \Gamma}{\mathrm{d} t} &= \frac{k}{\mu \sinh2\mu} \sin\Delta \epsilon (\cos 2\Gamma + \delta \sin^2\Gamma),\\
    \frac{\mathrm{d} \Delta\epsilon}{\mathrm{d} t} &= \frac{2k}{\mu \sinh2\mu} \left[\left(1-\frac{\delta}{2}\right)\cosh2\mu - \mu\sinh2\mu + \left(\frac{1}{\sin 2\Gamma} - \frac{\delta}{2}\tan\Gamma\right)\cos\Delta\epsilon \right].
\end{align}
\end{subequations}\end{linenomath*}

As was noted in \cite{Mak-et-al24}, the dynamical system \eqref{eq:amp_phase_ratios} has centers corresponding to cases where there are no exponentially growing normal modes with freely propagating Rossby waves that interact but do not phase-lock. The system can bifurcates into a case with one stable and one unstable fixed point when the marginal stability threshold is crossed; the \emph{stable} fixed point corresponds to the exponentially \emph{growing} normal mode, i.e. the CRW pairs gets \emph{attracted} into a phase-locked position that leads to mutual amplification. The bifurcation behavior is associated with that of $\mbox{PT}$ symmetry \cite[e.g.,][]{David-et-al22, MengGuo23}: the solutions initially have unbroken $\mbox{PT}$ symmetry with two real eigenvalues associated with waves with opposite Krein signatures (they propagate in different directions, cf. also positive/negative energy modes, e.g., \cite{Cairns79}), collide at exceptional points (the locations of marginal stability) and lead to solutions with spontaneously breaking of $\mbox{PT}$ symmetry with a complex conjugate pair of eigenvalues.

Note that two constants of motion can be identified from the above equations as (cf. those of \cite{DaviesBishop94} but up to overall multiplicative factors)
\begin{linenomath*}
\begin{align}\label{eq:const_of_motion}
  T^2 - \frac{B^2}{1-\delta} = \textnormal{constant}, \qquad     \frac{(T^2 + B^2)\cosh 2\mu + 2 TB \cos\Delta\epsilon}{\mu\sinh 2\mu} - \left(T^2 + \frac{B^2}{1-\delta}\right) = \textnormal{constant}.
\end{align}
\end{linenomath*}
These turn out to correspond precisely to that of pseudomomentum and pseudoenergy conservation. Further, one can define a Hamiltonian proportional to the pseudoenergy that has an action-angle formulation, where the action is proportional to the pseudomomentum and the angle relates to the phase; the interested reader is referred to Appendix~\ref{app:B}. If we assume $B \neq 0$, then dividing the two constants of motion by $B^2$ and supposing we are in the case with normal mode exponential growth and $T^2 / B^2 \to \textnormal{constant} < \infty$ leads to (using some trigonometric identities relating to $\tan\Gamma$ carefully, e.g., correcting a factor of 2 in the expression for $\sin 2\Gamma$ in \cite{Mak-et-al24})
\begin{linenomath*}\label{eq:NM_solu}
\begin{align}
  \frac{T}{B} = \tan\Gamma = \frac{1}{\sqrt{1-\delta}}, \qquad \cos\Delta\epsilon = \frac{2\mu\sinh 2\mu - (2-\delta)\cosh 2\mu}{2\sqrt{1-\delta}},
\end{align}
\end{linenomath*}
which are the analytical expressions relating to the graphs and marginal stability boundaries (where $|\cos\Delta\epsilon| = 1$) in Fig.~5 of \cite{Mak-et-al24}. These constants of motion can in principle be used to construct certain analytical solutions \cite[e.g.,][]{DaviesBishop94, KalashnikChkhetiani18}, although this was not pursued here.


\section{Characteristics of transient growth}\label{sec:transient}

The study of transient growth and the associated machinery is well-established and has been documented accordingly \cite[e.g.,][]{Trefethen-et-al93, FarrellIoannou96a, Schmid07, Kerswell18}. Here we only recap the portions we need for the present study of transient growth in the linear regime. The system \eqref{eq:A_equ} formally has a solution $\mathbf{q}(t) = \exp[\mbox{A}t]\mathbf{q}(0) = \mbox{P}(t)\mathbf{q}(0)$, where $\mbox{P}(t)$ is termed the \emph{propagator}. For studying transient growth the goal is to find, for some chosen time window $[0, t]$, the control state vector $\mathbf{q}(0)$ that maximizes the response
\begin{linenomath*}
\begin{align}\label{eq:gain}
  G(t) = \frac{\|\mathbf{q}(t)\|_{\mbox{M}}^2}{\|\mathbf{q}(0)\|_{\mbox{M}}^2} = \frac{\|\mbox{P}(t)\mathbf{q}(0)\|_{\mbox{M}}^2}{\|\mathbf{q}(0)\|_{\mbox{M}}^2}, \qquad \|\mathbf{a}\|_{\mbox{M}}^2 = \mathbf{a}^\dagger \mbox{M} \mathbf{a} = (\mbox{M}^{1/2}\mathbf{a})^\dagger (\mbox{M}^{1/2}\mathbf{a}).
\end{align}
\end{linenomath*}
Here $(\cdot)^\dagger$ denotes a Hermitian conjugate (component-wise complex conjugate and a transpose for matrices), $\mbox{M}$ is a positive-definite Hermitian matrix, so that $\mbox{M}^{1/2}$ is well-defined and is also Hermitian (invoked above as $(\mbox{M}^{1/2})^\dagger = \mbox{M}^{1/2}$), and $\|\cdot\|_{\mbox{M}}$ is the `standard' norm in the transformed variable $\tilde{\mathbf{q}} = \mbox{M}^{1/2}\mathbf{q}$. The resulting quantities from the optimization depend ultimately on the choice of optimization window $[0, t]$ and norm implied by $\mbox{M}$.

In the linear case it can be shown that \cite[e.g.,][]{FarrellIoannou96a, Schmid07} the maximizer can be obtained by performing a singular value decomposition of $\tilde{\mbox{P}}(t) = \mbox{M}^{1/2} \mbox{P}(t) \mbox{M}^{-1/2} = \mbox{U}\Sigma\mbox{V}^\dagger$, where $\mbox{U}$ and $\mbox{V}$ are unitary (i.e. $\mbox{U}^\dagger\mbox{U} = \mbox{I}$), and $\Sigma$ consists of the singular values of the propagator. An object of focus here is the right singular vector $\mathbf{v}_1 = \mbox{M}^{1/2}\mathbf{q}_{\rm opt}(0)$ taken from $\mbox{V}$ associated with the largest singular value $\sigma_1$; the left singular vector associated with $\mbox{U}$ provides the end state proportional to $\mathbf{q}_{\rm opt}(t)$ where evolution is governed here by $\mbox{A}$. In the $t \to \infty$ limit it is the normal model solutions woth exponential growth (where they exist) that dominate, and this we compute through finding the eigenvalue and eigenvectors of $\mbox{A}$ (or $\tilde{\mbox{A}} = \mbox{M}^{1/2} \mbox{A} \mbox{M}^{-1/2}$, which will have the same eigenvalues since it is a self-similar transform). Note this is not a statement about the normal mode eigenfunctions and corresponding optimal finite time initial conditions coinciding in the $t \to \infty$ limit, as will be discussed briefly in Sec.~\ref{sec:finite_time}. In all the above cases, we can also compute the associated CRW quantities associated with the optimal initial perturbation $\mathbf{q}_{\rm opt}(0) = (q_T, q_B)$ defined by
\begin{linenomath*}
\begin{align}\label{eq:opt_deps_gamma}
  \Delta \epsilon = \mbox{arg}\left(\frac{q_T}{q_B}\right) = \epsilon_T - \epsilon_B, \qquad \Gamma = \arctan \left(\frac{|q_T|}{|q_B|}\right) = \frac{T}{B}.
\end{align}
\end{linenomath*}

As concrete examples of choice of norms, the trivial case of no transformation corresponds to a maximization under the (potential) enstrophy norm, i.e.
\begin{linenomath*}
\begin{align}\label{eq:ens_norm}
  \mbox{M} = \mbox{I} \qquad \Leftrightarrow \qquad Z = \left\langle q, q\right\rangle = \frac{1}{2} \int |q|^2 \; \mathrm{d}V,
\end{align}
\end{linenomath*}
appropriately projected onto the CRW variables, because in our case the equations of motion described by $\mbox{A}$ is in terms of potential vorticity variables, and we have introduced the inner produce notation $\langle a, b \rangle = (1/2) \int a^* b\; \mathrm{d}V$ notation. On the other hand, it can be shown that (cf. the procedure in \cite{HeifetzMethven05})
\begin{linenomath*}
\begin{align}\label{eq:ene_norm}
  \mbox{M} = \frac{1}{\mu\sinh 2\mu}
    \begin{pmatrix}
      \cosh 2\mu & 1 \\
      1 & \cosh 2\mu 
    \end{pmatrix} 
    \qquad \Leftrightarrow \qquad E = \left\langle \nabla\psi + \frac{\partial\psi}{\partial z}, \nabla\psi + \frac{\partial\psi}{\partial z} \right\rangle  = -\left\langle \psi, q \right\rangle,
\end{align}
\end{linenomath*}
which corresponds to the total (kinetic and potential) energy norm after performing an integration by parts and projecting onto the CRW variables; the inner product defined in terms of $\psi q$ leads to easier manipulations, as the full form results in some integrals in depth for the hyperbolic sine and cosine terms from \eqref{eq:psi_wave}, although the expressions returned are identical in both cases. There are other choices of norms one could consider (e.g., kinetic energy, potential energy, pressure, or more arbitrary ones in principle), although here we will only present results for enstrophy and (total) energy, and only mention the results for the other cases (the kinetic, available potential energy and boundary pressure norms, as was considered in \cite{KalashnikChkhetiani18}) as they come up in the narrative, since the two presented here seem to display the representative behavior in the set of norms we have considered (results for kinetic and boundary pressure norm are qualitatively similar to energy norm, while the available potential energy norm is qualitatively similar to enstrophy norm). The interested reader can consult Appendix~\ref{app:C} or the analysis code provided for results associated with the other norms stated above.


\subsection{Instantaneous growth}\label{sec:inst_time}

One limit where some simplification could be made is the infinitesimal time limit $t \ll 1$. Working in the transformed variable $\tilde{\mathbf{q}}$, note that \eqref{eq:gain} is
\begin{linenomath*}
\begin{align*}
  \|\mbox{P}(t)\mbox{q}(0)\|_{\mbox{M}}^2 = \mathbf{q}^\dagger(0) \mbox{P}^\dagger(t) \mbox{M} \mbox{P}(t) \mathbf{q}(0) = \tilde{\mathbf{q}}^\dagger(0) \tilde{\mbox{P}}^\dagger(t) \tilde{\mbox{P}}(t) \tilde{\mathbf{q}}(0),
\end{align*}
\end{linenomath*}
where $\tilde{\mbox{P}} = \mbox{M}^{1/2} \exp[\mbox{A}t] \mbox{M}^{-1/2}$, and we have again made use of the fact that $\mbox{M}^{\pm 1/2}$ are Hermitian (since $\mbox{M}$ is Hermitian and positive-definite). We have
\begin{linenomath*}
\begin{align}
  \tilde{\mbox{P}}^\dagger \tilde{\mbox{P}} = \mbox{M}^{-1/2}\ (\exp[\mbox{A}^\dagger t]\ \mbox{M}\ \exp[\mbox{A} t])\ \mbox{M}^{-1/2},
\end{align}
\end{linenomath*}
and a Taylor expansion of the matrix exponential for $t\ll 1$ gives
\begin{linenomath*}
\begin{align}\label{eq:A_tilde_formulation}
  \tilde{\mbox{P}}^\dagger \tilde{\mbox{P}} \approx \mbox{M}^{-1/2}\ [\mbox{M} + (\mbox{MA} + \mbox{A}^\dagger\mbox{M})t + O(t^2)]\ \mbox{M}^{-1/2} = \mbox{I} + [\tilde{\mbox{A}} + \tilde{\mbox{A}}^\dagger] t + O(t^2)
\end{align}
\end{linenomath*}
where $\tilde{\mbox{A}} = \mbox{M}^{1/2}\mbox{A}\mbox{M}^{-1/2}$ is the governing operator under a self-similarity transform associated with the choice of norm defined by $\mbox{M}$. For $\tilde{\mbox{S}} = [\tilde{\mbox{A}} + \tilde{\mbox{A}}^\dagger]/2$, the leading order behavior can be obtained from an eigenvalue analysis of $\tilde{\mbox{S}}\tilde{\mathbf{q}}(0) = \sigma\tilde{\mathbf{q}}(0)$, with the largest eigenvalue associated with the largest instantaneous growth in the chosen norm, and the eigenvector be the relevant optimal solution \cite[e.g.,][]{FarrellIoannou96a}. Alternatively, we could consider staying in the $\mathbf{q}$ variables, which would lead to
\begin{linenomath*}
\begin{align}\label{eq:M_formulation}
  \mbox{P}^\dagger \mbox{M} \mbox{P} \approx \mbox{M} + [\mbox{M} \mbox{A} + \mbox{A}^\dagger \mbox{M}]t + O(t^2),
\end{align}
\end{linenomath*}
and in this case we would perform a \emph{generalized} eigenvalue analysis $\mbox{S}\mathbf{q}(0) = \sigma\mbox{M}\mathbf{q}(0)$ where $\mbox{S} = [\mbox{M} \mbox{A} + \mbox{A}^\dagger \mbox{M}]/2$. We think the formulation in terms of $\tilde{\mbox{S}}$ is more intuitive (choose a norm, transform the variables and governing operator, and find eigenvalues of the transformed operator added to its Hermitian conjugate), but the formulation in terms of $\mbox{S}$ has the convenience that we do not need to deal with the square root of a matrix. The numerical results for this work employs the formulation with $\tilde{\mbox{S}}$, although one of the analytical results is more conveniently derived with the formulation in $\mbox{S}$.

For comparison purposes with previous work \cite[e.g.,][]{Mak-et-al24} we focus on the meridional wavenumber $l=0$ case (and so $\mu = k$ in the above), and we comment on the $l > 0$ cases in the conclusion section. In Fig.~\ref{fig:fig2} we show for each of the $\delta$ positive, zero and negative case the instantaneous growth rate associated with the largest eigenvalue value of $\tilde{\mbox{S}}$ in the enstrophy and (total) energy norm (red and green line respectively); we also show the exponentially growing normal mode instability (as $\sigma = \omega_i = kc_i$; black dashed line) for reference. Consistent with previous works albeit for different systems \cite[e.g.,][]{HeifetzMethven05, KalashnikChkhetiani18}, the instantaneous growths surpass that of the normal mode grwoth across the bandwidth (which is in fact the case for all norms considered in this work), and is generally non-zero even when there is no normal mode instability with exponential growth, although of course the growth rate is only for an infinitesimal time horizon. Like those previous works, the maximum instantaneous growth rates for enstrophy and energy norms occur at the small $k$ values, with the enstrophy one approaching infinity with $k\searrow 0$. One not so obvious feature is that the instantaneous growth in the energy norm seems to be $\delta$-independent, which is not the case for the enstrophy norm displayed here (the red curve in Fig.~\ref{fig:fig2} has a shift to the `right' as we move down the panels), nor is it true for any of the other norms considered; we will come back to this observation shortly.

\begin{figure}
  \centering
  \includegraphics[width=0.5\linewidth]{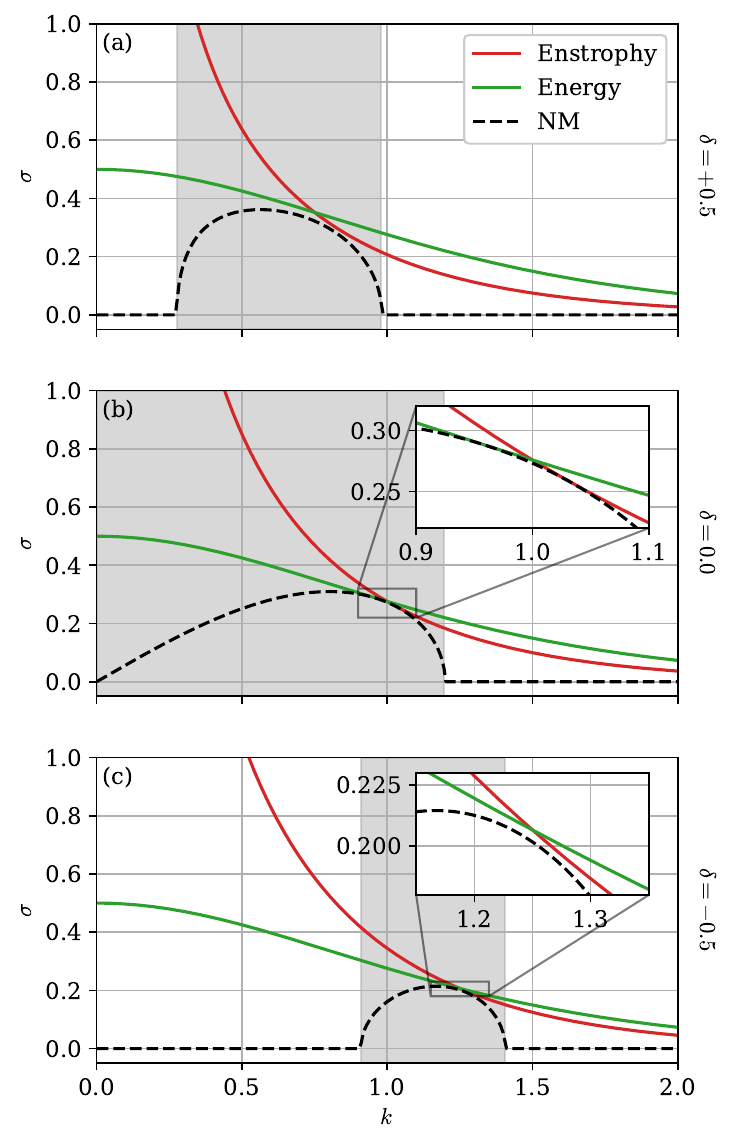}
  \caption{Comparison between the instantaneous growth rates in the enstrophy and energy norms (red and green respectively) and the growth rates of normal modes (black dashed) as a function of zonal wavenumber $k$ (taking meridional wavenumber $l=0$), for three choices of slope parameter $\delta$. The shaded regions indicate the region where there is normal mode instability.}
  \label{fig:fig2}
\end{figure}

In the $\delta = 0$ case (i.e. no slope, standard Eady problem) it was determined that the instantaneous growth curves do intersect and coincide with the normal mode growth at exactly one $k \in (0, \infty)$ for the present case where $l=0$, determined numerically to be $k=1.031$ for the enstrophy norm and $k=0.955$ for the energy norm. This observation relates to the normality of the operator $\mbox{A}$ and the appropriate $\tilde{\mbox{A}}$ at those specific values of $k$. When the matrix of the operator is normal (i.e. $\mbox{A}\mbox{A}^\dagger = \mbox{A}^\dagger\mbox{A}$) the magnitudes of the singular values and eigenvalues coincide (and the singular vectors and eigenvectors differ by a phase). On the other hand, it would appear that when $\delta \neq 0$ there is no such agreement between the computed growth rates at least for the cases examined, suggesting the corresponding operator is generically non-normal. The exact value of $k$ where the $\mbox{A}$ (or $\tilde{\mbox{A}}$) operator is normal or not can in principle be computed by hand via examining the components of $\mbox{A}\mbox{A}^\dagger - \mbox{A}^\dagger\mbox{A}$, possibly finding that for $\delta \neq 0$ there is no $k$ value where all components vanish identically. In our case we opted instead to computed a matrix norm numerically, and found they were numerically non-zero (we chose the $L_{1,1}$ norm, i.e. sum of the absolute values of the matrix entries).

In Fig.~\ref{fig:fig3} we show the CRW quantities associated with the instantaneous growth in the enstrophy and energy norms (left and right columns) for the same three choices of slope parameter $\delta$, with the associated quantities for the normal mode for comparison. We first review the results for the normal mode cases as reported in \cite{Mak-et-al24}, since we will revisit this later for instability characteristics for finite time transient growth in Sec.~\ref{sec:finite_time}. The phase-shift $\Delta\epsilon$ is $\pi$ (out of phase) when $k$ is small and there is no exponential growth, and is $0$ (in phase) when $k$ is large and also no exponential growth, with a transition region in between. Note that the optimal exponential growth occurs at $\Delta\epsilon \gtrsim \pi/2$ rather than exactly at $\pi/2$; physically, the mutual hindering effect is desirable for phase-locking of the CRW pairs. On the other hand, the normal mode CRW amplitude ratios are constant throughout the region with exponential growth, with $\Gamma$ being either side of $\pi/4$ depending on the choice of $\delta$ (with $\Gamma = \pi/4$ at $\delta = 0$); recalling that $\Gamma = \arctan(T/B)$, the bottom CRW is \emph{stronger} for $\delta < 0$ while it is \emph{weaker} for $\delta > 0$, consistent with the choice of $\delta$ strengthening or weakening the bottom PV gradient respectively.

\begin{figure}
  \centering
  \includegraphics[width=0.9\linewidth]{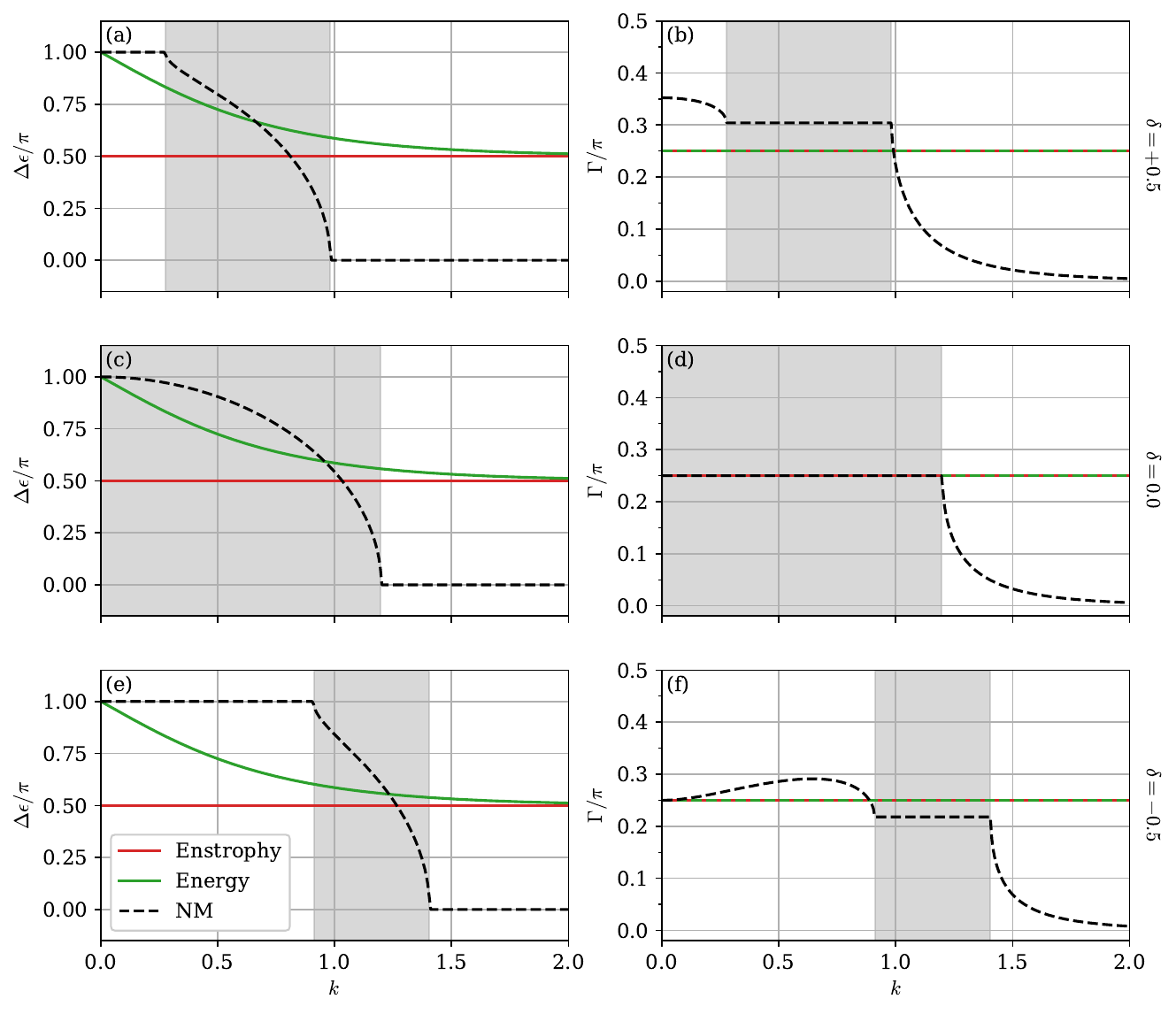}
  \caption{CRW quantities associated with the instantaneous growths in enstrophy and energy norm (red and green respectively) and for the normal mode (black dashed) as a function of zonal wavenumber $k$ (taking meridional wavenumber $l=0$), for three choices of slope parameter $\delta$. ($a,c,e$) The scaled phase-shift $\Delta\epsilon / \pi$. ($b,d,f$) Scaled amplitude ratio $\Gamma / \pi$. The shaded regions indicate the region where there is normal mode instability.}
  \label{fig:fig3}
\end{figure}

With the observations and its dependence on $\delta$ for the normal mode in mind, it is perhaps interesting to see that, for instantaneous growth, $\Delta\epsilon$ and $\Gamma$ seem to both be $\delta$-independent! The amplitudes of the CRWs of instantaneously growth appears to always be equal ($\Gamma = \pi/4$) regardless of choice of norm and $\delta$. On the other hand, the phase-shift $\Delta\epsilon$ is $\pi/2$ for the enstrophy norm (consistent with the works of \cite{HeifetzMethven05, KalashnikChkhetiani18} for the Rayleigh sheet and non-sloped Eady system), but is something between $\pi$ and $\pi/2$ for the energy norm dependent on $k$ not independent of $\delta$. This may suggest the $k$-dependence is a manifestation of the transformation into different norms by the choice of $\mbox{M}$ in \eqref{eq:ene_norm}, but the $\delta$-independence may be something intrinsic to the system governed by $\mbox{A}$ or its transformations. For all norms considered, where $\Delta\epsilon$ is not equal to $\pi/2$, in the large $k$ limit they do tend to $\pi/2$ from above (the available potential energy case being the exception, which approaches $\pi/2$ from below).

In fact, we have the following two technical results:
\begin{prop}\label{prop:1}
  For the problem governed by $\mbox{\textnormal{A}}$ given in \eqref{eq:A_equ}, $\Delta\epsilon$ and $\Gamma$ given in \eqref{eq:opt_deps_gamma} associated with the maximum instantaneous growth in some norm defined by $\mbox{\textnormal{A}}$ is independent of the slope parameter $\delta$ if $\mbox{\textnormal{M}}$ is $\delta$-independent, real and symmetric, as well as having equal entries along the main diagonal.
\end{prop}

\begin{proof}
  Consider first the trivial case with $\mbox{M} = \mbox{I}$, i.e. the (potential) enstrophy norm. Observe that
  \begin{linenomath*}
  \begin{align*}
    \mbox{A} + \mbox{A}^\dagger = \alpha(k,l,\delta) 
  \begin{pmatrix}
  0 & 1 \\
  -1 & 0
  \end{pmatrix} = \alpha(k,l,\delta) \mbox{J}, \quad
  \alpha(k,l,\delta) = \mathrm{i} k \frac{2-\delta}{\mu \sinh(2\mu)},
  \end{align*}
  \end{linenomath*}
  where $\mbox{J}$ is the canonical $2\times 2$ symplectic matrix. The eigenvalue problem is then $\alpha(k, l, \delta) \mbox{J} \mathbf{q} = \sigma \mathbf{q}$, but because of the symplectic structure, the non-degenerate case with $\sigma > 0$ has the linear system reducing down to $(q_T / q_B)^2 = -1$, and there is no $\delta$-dependence remaining since it gets divided out for $\alpha \neq 0$. For this simple form of $\mbox{M}$ we have $\Delta\epsilon = \mbox{arg}(\pm \mathrm{i}) = \pm \pi/2$ and $\Gamma = \arctan|\pm \mathrm{i}| = \pi / 4$.
  
  For general $\mbox{M}$, note that $\mbox{M}^{\pm 1/2}$ is well-defined and Hermitian, since $\mbox{M}$ is assumed to be positive-definite Hermitian matrix. A general form of a $2\times 2$ Hermitian matrix has the form
  \begin{linenomath*}
  \begin{align*}
    \mbox{M}^{1/2} = 
      \begin{pmatrix}
        a & c \\ c^* & b
      \end{pmatrix}
  \end{align*}
  \end{linenomath*}
  where $a,b$ are real but $c$ is complex, and in general these can be functions of $k$, $l$ and $\delta$ also. However, it can be shown (e.g., brute force computation) that
  \begin{linenomath*}
  \begin{align*}
    \tilde{\mbox{A}} + \tilde{\mbox{A}}^\dagger = \mbox{M}^{1/2}\mbox{A}\mbox{M}^{-1/2} + \mbox{M}^{-1/2}\mbox{A}^\dagger\mbox{M}^{1/2} = \alpha_{\mbox{M}}(k, l, \delta) \mbox{J}
  \end{align*}
  \end{linenomath*}
  only if $a = b$ and $c \in \mathbb{R}$, for some other function $\alpha_{\mbox{M}}$ that depends on $\mbox{M}$. The associated linear system for $\alpha_{\mbox{M}} \neq 0$ reduces to $(\tilde{q}_T / \tilde{q}_B)^2 = -1$. Reverting to the CRW variables by $\mathbf{q} = \mbox{M}^{-1/2}\tilde{\mathbf{q}}$ and assuming that $\mbox{M}^{-1/2}$ is $\delta$-independent, the resulting $\Delta\epsilon$ would still be $\delta$-independent, although it may have other dependence arising from $\mbox{M}$ (e.g., wavenumber dependence, as in the energy norm considered in Eq.~\ref{eq:ene_norm}). Since $\tilde{q}_T = \pm \mathrm{i}\tilde{q}_B$, we have (taking $a=b$ and $c\in\mathbb{R}$)
  \begin{linenomath*}
  \begin{align*}
    \frac{|q_T|}{|q_B|} = \frac{T}{B} = \frac{|-c \pm\mathrm{i}a|}{|a \mp\mathrm{i}c|} = 1,
  \end{align*}
  \end{linenomath*}
  and so we still have $\Gamma = \pi / 4$ in the CRW variables for this system. 
  
  The same argument goes through if we instead consider looking at $2\mbox{S} = \mbox{M}\mbox{A} + \mbox{A}^\dagger \mbox{M}$.
\end{proof}

\begin{prop}\label{prop:2}
  For the energy norm, all properties associated with the instantaneous growth (e.g., the instantaneous growth rate as seen in Fig.~\ref{fig:fig2}) is $\delta$-independent.
\end{prop}

\begin{proof}
    For this case it is easier to deal with $\mbox{M}$ instead of $\mbox{M}^{1/2}$ through looking at $2\mbox{S} = \mbox{M}\mbox{A} + \mbox{A}^\dagger\mbox{M}$. Restricting our attention to norms with a $\mbox{M}^{1/2}$ with the same properties as in the previous proposition, then
    \begin{linenomath*}
    \begin{align*}
      \mbox{M}^{1/2} = 
        \begin{pmatrix}
          a & b \\ b & a
        \end{pmatrix} \qquad \Leftrightarrow \qquad
      \mbox{M} = 
        \begin{pmatrix}
          c & d \\ d & c
        \end{pmatrix}
    \end{align*}
    \end{linenomath*}
    where all entries are real by assumption. Computing the entries of $2\mbox{S}$ we find it takes the symplectic form as in the previous result, but with
    \begin{linenomath*}
    \begin{align*}
      \alpha_{\mbox{M}}(k, l, \delta) = \mathrm{i}k\left[2d + (c-d \cosh2\mu)\frac{2-\delta}{\mu\sinh 2\mu}\right],
    \end{align*}
    \end{linenomath*}
    so if $d=g$ and $c=g\cosh 2\mu$ up to some factor $g$ --- which is the case for the energy norm given in \eqref{eq:ene_norm} --- the resulting operator is $\delta$-independent. 
    
    The above argument goes through if we instead consider look at $2\tilde{\mbox{S}} = \tilde{\mbox{A}} + \tilde{\mbox{A}}^\dagger = \mbox{M}^{1/2}\mbox{A}\mbox{M}^{-1/2} + (\mbox{M}^{1/2}\mbox{A}\mbox{M}^{-1/2})^\dagger$, but the algebra is substantially more tedious.
\end{proof}

Physically, this $\delta$-independence in the CRW quantities for instantaneous growth may be rationalized in that, for the infinitesimal time limit, we are only interested in the configuration of the initial perturbation that leads to growth, and not at all on the mutual interaction, which is impacted by $\delta$. In that sense $\Delta\epsilon = \pi/2$ is something that we might expect already: since $v' = \partial\psi/\partial x = \mathrm{i}k \psi$ and $q \sim -\psi$, $v$ and $q$ is out of phase by $\pi/2$, and for growth in the enstrophy norm we essentially want maximum correlation between $q$ at one boundary from the CRW at that boundary with the implied $v$ at the same boundary arising from the CRW at the other boundary, which occurs at $\pi/2$. It is however interesting and not obvious \emph{a priori} that the energy norm has this property where the governing operator is $\delta$-independent. 

We should note that it is also entirely possible that maximum instantaneous growth has a subsequent growth that leads to large decay. The physical relevance and applicability in that sense depends ultimately on the time horizon and the choice of norm, which is of course consistent with the fact that transient growth as formulated is really one of an optimization problem, with solutions affected by the user defined objective or loss function and is context dependent.


\subsection{Finite time transient growth}\label{sec:finite_time}

We already know that the normal modes (which manifest and become the dominant feature in the $t\to\infty$ limit) are affected by the choice of $\delta$, so in this case we would expect that the transient growth to be affected by $\delta$, albeit increasingly less so for small time horizons. Here we demonstrate similar growth-decay effects for the instability characteristics of finite time transient growth. For this we first need to either derive analytical expressions for the matrix exponential $\tilde{\mbox{P}}(t) = \mbox{M}^{1/2}\exp[\tilde{\mbox{A}}t]\mbox{M}^{-1/2}$ (which can be done for simple cases, e.g., \cite{HeifetzMethven05} for the Rayleigh problem with $\mbox{M} = \mbox{I}$), or choose a time $t$ and numerically evaluate the matrix exponential; we take the latter approach here for simplicity. We compute the singular value decomposition (also done numerically here) and evaluate an effective growth rate as (cf. \cite{HeifetzMethven05})
\begin{linenomath*}
\begin{align}\label{eq:sig_eff}
  \sigma_{\rm eff}(t) = \frac{1}{t}\log \sigma_1(t),
\end{align}
\end{linenomath*}
where $\sigma_1(t)$ is the largest singular value from the singular value decomposition of the propagator for some choice of $t$. The computed growth rates for the three choices of $\delta$ considered so far for both the enstrophy and energy norm are shown in Fig.~\ref{fig:fig4}.

\begin{figure}
  \centering
  \includegraphics[width=0.9\linewidth]{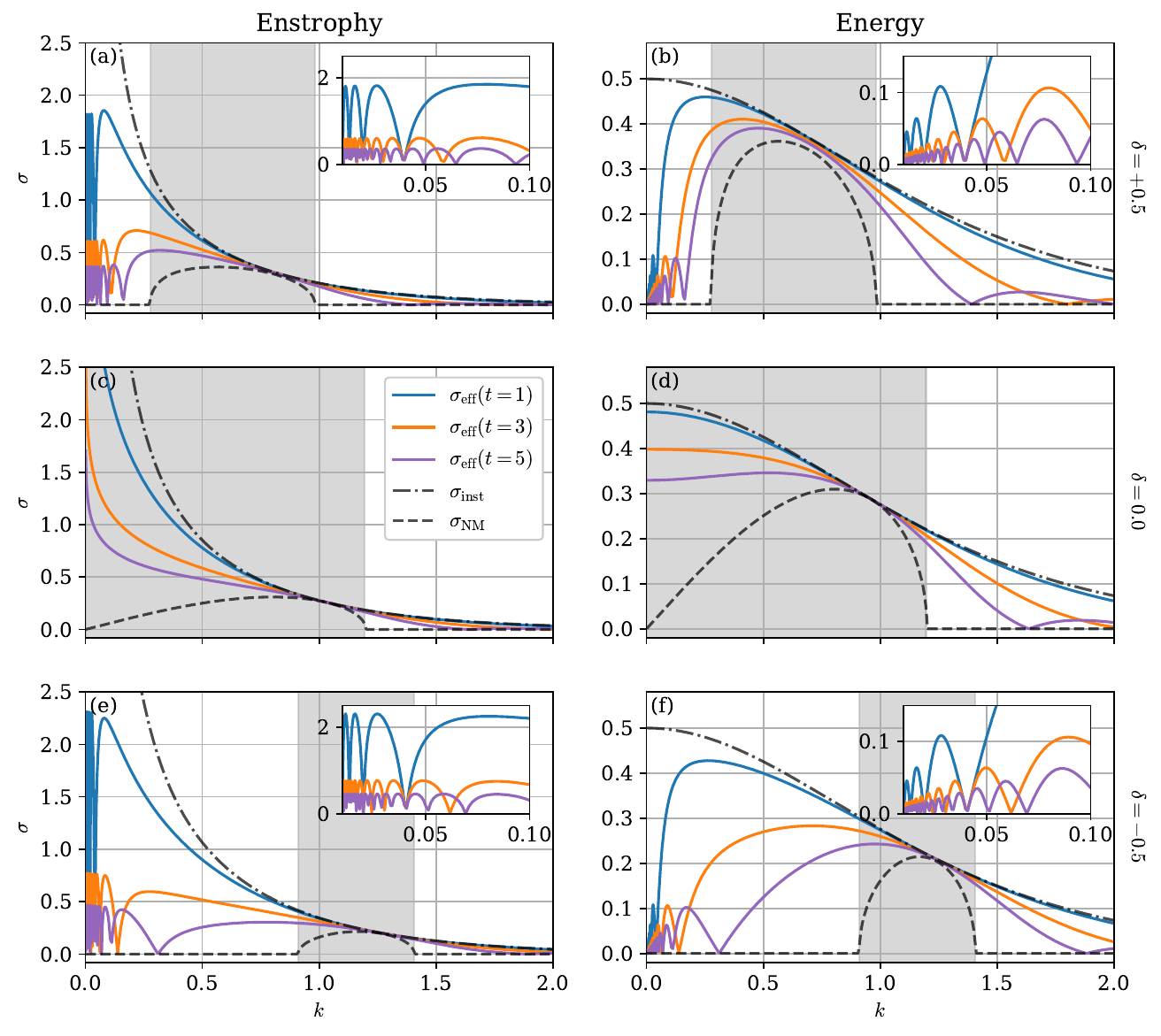}
  \caption{The effective optimal growth rate $\sigma_{\mathrm{eff}}(t)$, as a function of wavenumber $k$ (taking meridional wavenumber $l=0$) for three choices of $\delta$, for the (non-dimensional) target times $t=1$, $3$, $5$ (blue, orange and purple respectively), with insets zoomed in towards the small $k$ values where relevant. The normal mode growth rates and instantaneous growth rates are plotted for comparison purposes (black dashed and dot-dashed respectively), and the shaded regions indicate the region where there is normal mode instability.}
  \label{fig:fig4}
\end{figure}

Possibly a first thing to notice is that there is a `bouncing' of the finite time growth rates that is present for $\delta \neq 0$ for small $k$ values that is  absent for the $\delta = 0$ case. This is in fact not so much to do with the choice of $\delta$, but more a phenomenon to do with whether we are in the normal mode unstable bandwidth (a hint of this could be seen in panel $b,d,f$ for $\sigma_{\rm eff}(t=5)$; purple line). The bouncing behavior can be rationalized as the following: in the absence of normal mode instability, the CRWs do not phase-lock and are freely propagating, and they can mutually interact, which leads to growth when $\Delta\epsilon\in(0, \pi/2)$, but can decay when $\Delta\epsilon\in(-\pi/2, 0)$. The bouncing effect then simply comes from the choice of $t$ leading to some averaging of the growth and decay phases as the CRW pair cycle through the $\Delta\epsilon$ as they freely propagate, and is an artifact from the definition of $\sigma_{\rm eff}$; for example, $\sigma_{\rm eff}(t) = 0$ for some appropriately chosen $t$ could encompass a case where there is equal growth and decay from the transient interaction. For $k\searrow 0$ the `bouncing' occurs more frequently (see for example the insets in Fig.~\ref{fig:fig4}), which relates to the fact that CRWs here have a phase speed that goes like $k^{-1}$, so the CRWs propagate faster, and more cycle of growth and decay occur. An analogous decrease in the `bouncing' can be seen for $k\gg 1$ and relates to the waves becoming slower, so needing a longer time to complete the cycle (not shown here, but see for example that $t=5$ has one zero crossing on the right of the shaded region while the small $t$ cases do not). Similar behavior for example in the phase-shifts $\Delta\epsilon$ in the normal mode stable regions as `sawtooth' patterns (not shown here, but see e.g. Fig.~6$b$ of \cite{HeifetzMethven05}).

In all cases considered the finite time transient growth rates are bounded above by the instantaneous growth rate, and below by the normal mode growth rate. There is equality only when the governing operator $\tilde{\mbox{A}}$ or $\tilde{\mbox{A}}$ is normal, which is only the case when $\delta = 0$ for the numerically determined values values of $k$ reported in the previous subsection (closer examination of Fig.~\ref{fig:fig4}$a,e,f$ indicate there are no intersections of the curves). In the $k\to\infty$ setting the `bouncing' in the growth rate will have to decrease in magnitude since the instantaneous growth rate tends to zero in that limit. In the limit of $k\searrow0$, however, such bound does not hold, and while the `bouncing' does decrease in magnitude for the energy norm, there is some finite amplitude asymptote in the enstrophy norm for the $\delta \neq 0$ cases. In all cases for $\delta$, there does appear to be some sort of convergence of the finite time transient growth towards the normal mode growth in that the former's general amplitude is a decreasing function of $t\to\infty$ for all norms considered except in the enstrophy norm for $\delta=0$, where there seems to be some sort of singular behavior that leads to non-uniform convergence to the normal mode growth rate at $k\searrow0$.

As part of the singular decomposition for analyzing finite time transient growth we obtain the optimal initial conditions $\mathbf{q}(0; k, l, \delta, t) = (q_T(0), q_B(0))$, and we can compute the associated CRW properties. We show in Fig.~\ref{fig:fig5} the CRW properties associated with these optimal initial conditions as a function of varying $t$ on the appropriate phase portraits in $(\Delta\epsilon, \Gamma)$ phase space, obtained numerically from \eqref{eq:amp_phase_ratios} for sample but representative choices of chosen $\delta$ and $k$ ($\delta=+0.5, 0.0, -0.5$, with $k$ corresponding to the optimal normal mode growth of $\delta=+0.5, 0.0, -0.5$). Recalling the behavior of the phase portraits first, for the $\delta=0$ case (i.e. no slope and standard Eady problem, middle row), there is normal mode instability for these three choices of $k$, and the attractor of this system (denoted by the black circle) lies on $\Gamma = \pi/4$, i.e. top and bottom CRW have equal magnitude. With increasing $k$, the two fixed points come increasingly close towards $\Delta\epsilon = 0$, collide and bifurcate into two centers. For $\delta < 0$ (bottom row), the case where the topographic slope strengthens the bottom PV gradient, where there is normal mode instability it is such that $\Gamma < \pi/4$, i.e. the bottom wave is stronger (e.g., panel $i$), and vice-versa for $\delta > 0$ (top row, e.g., panels $a,b$). If there is no normal mode instability, then depending on whether we are to the left or right of the normal mode instability wavenumber bandwidth, we have centers lying on the line $\Delta\epsilon = \pm \pi$ (e.g., panels $g,h$) or $\Delta\epsilon = 0$ (e.g., panel $c$) respectively. In all cases, the optimal normal mode instability for fixed $\delta$ over all values of $k$ (panels $a,e,i$) is such that $\Delta\epsilon > \pi/2$, even if $\Delta\epsilon = \pi/2$ is the optimally interfering configuration (it is because we want to be in a \emph{hindering} configuration for phase-locking; see \cite{Mak-et-al24} for more details).

\begin{figure}
    \centering
  \includegraphics[width=1.0\linewidth]{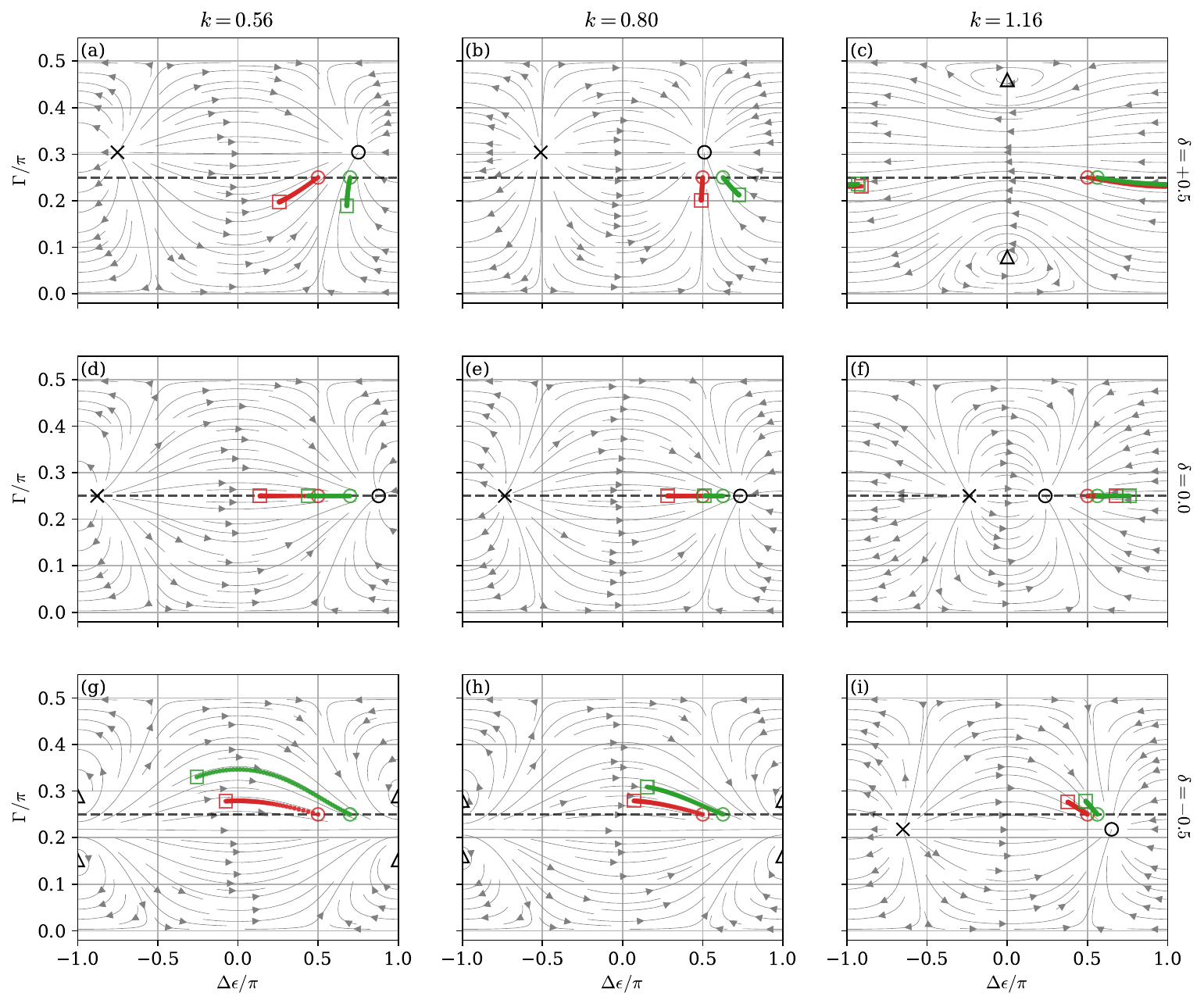}
  \caption{Phase portrait in ($\Delta\epsilon, \Gamma$) phase space and the fixed points of the dynamical system, for three values of $\delta$ (rows, $\delta=+0.5, 0.0, -0.5$ from top to bottom) and three choices of $k$ (columns, $k$ corresponding to the optimal normal mode growth of $\delta=+0.5, 0.0, -0.5$ from left to right). Stable and unstable fixed points (marked on as a circle and cross respectively) correspond to the growing and decaying normal mode respectively, and when there is no normal mode growth the fixed points are centers (marked on as triangles). The finite-time optimization with respect to the enstrophy and total energy are put on as markers (red and green respectively) for target times $t\in [0.001,5]$ (circles and squares denote the optimal initial states corresponding to the shortest and longest target times respectively). When $\Gamma = \pi/4$ the top and bottom CRW have equal amplitude, and is marked on as a horizontal gray dashed line for reference.}
  \label{fig:fig5}
\end{figure}

For the finite time transient instabilities, starting first with the $\delta = 0$ case (middle row), we see that all such optimal initial conditions (for all norms examined for this work) lie on the line $\Gamma = \pi/4$ also. This is perhaps to be expected, since in the $\delta=0$ case there is no system asymmetry between the top and the bottom wave as such, so there is no reason for the optimal initial condition to not have equal magnitudes. The increasing distance from the attractor and the instantaneous growth (noting the enstrophy norm one starts at $\Delta\epsilon = \pi/2$ from Proposition~\ref{prop:1} and Fig.~\ref{fig:fig3}$b$) is because the solution's evolution in phase space is governed \eqref{eq:amp_phase_ratios} and confined to be on the associated streamlines, so the choice of increasing $t$ for finite transient growth might be expected to be akin to back-tracing from the initial condition of the instantaneous growth in phase space along the streamline accordingly. With those observations in mind, we find that, for $\delta \neq 0$, if the normal mode instability is on one side of the $\Gamma = \pi/4$ line, the finite time optimal initial conditions (for all norms considered) always wants to be on the \emph{other} side of that $\Gamma = \pi/4$ line, i.e. in the $\delta<0$ case the normal mode has bottom CRW that is stronger, but for finite time transient instabilities we want to start with a condition that has a \emph{weaker} bottom CRW, and vice-versa for $\delta > 0$. In the cases where there is no normal mode instability, the finite time optimal instabilities still want to have the same ordering for the values of $t$ shown (weaker bottom CRW for $\delta < 0$ and vice-versa). This may be reasonable physically and from an optimization point of view: such an initial configuration allows one to get more gain in the growth by exciting a CRW over where the background PV gradient is stronger (although presumably norm dependent, e.g., if we define a norm that does not consider the quantity evaluated at the bottom boundary, then we would probably want the top CRW to be weak relative to the bottom CRW instead). A further observation we make based on numerical evidence seems to be that the phase-shift $\Delta\epsilon$ of the finite time optimal instability for a chosen $t$ is larger in the energy norm than the enstrophy norm, and we have not found cases where the finite time optimal initial solutions are non-unique; these properties could possibly be shown analytically (cf. extending the analytical approaches given in \cite{KalashnikChkhetiani18}), but is not something we have pursued.

One point we do want to make is regarding a subtlety with the $t\to\infty$ limit. For the cases where there is no normal mode instability (Fig.~\ref{fig:fig5}$c,g,h$), taking larger target times $t$ leads to `arcs' being traced out in the $(\Delta\epsilon, \Gamma)$ phase space, with one end point at the point in phase space corresponding to instantaneous growth, and the other on the $\Gamma = \pi/4$ line along the same streamline that the instantaneous growth initial condition is on (hence our rather specific choice of words in the previous paragraph). These effectively correspond again to freely propagating CRWs that has growth and decay phases that may offset each other (cf. Fig.~\ref{fig:fig4} with the `bounces' in the growth rate), and a choice of finite $t$ aims to find the configuration with the largest overall growth. On the other hand, for the cases where there are normal mode instabilities (Fig.~\ref{fig:fig5}$a,b,d,e,f,i$), the finite time optimal given by $t=5$ is effectively representative of the $t \gg 1$ calculations, which does \emph{not} coincide with the location of the normal mode instability in phase space. This is not a contradiction, just that some care is needed in the relevant statements. Suppose for some finite but large $t$ and a given norm, the normal mode configuration has some rate of gain based on some choice of norm over the time window $[0, t]$. Over the same time window, the optimal finite time configuration will have some other gain over the time window $[0, t_s] \cup [t_s, t]$, where $t_s$ is some time where the configuration is `close' (in a sense that could be formalized properly) to the normal mode one as the finite time optimal CRW configuration evolves towards the normal mode instability CRW configuration. While the rate of gain of the finite time optimal configuration is equal to that of the normal mode in the latter time window, there will be \emph{additional} gain in the former time window $[0, t_s]$ in the finite time case (otherwise it would not be the optimal solution). The finite time transient growth rates may tend to the normal mode growth rates, but it does not mean the corresponding eigenfunction or initial conditions converge to the other (abusing notation somewhat, the $t=\infty$ and $t\to\infty$ settings may have the same growth rates, but they need not correspond to the same solutions).

For completeness, we show samples of the streamfunction (which is proportional to pressure) associated with optimal finite time transient growth in Fig.~\ref{fig:fig6} (in this case for target time $t=5$ for the enstrophy norm): we construct the $\hat{\psi}_{T,B}$ from $\hat{q}_{T,B}$ using \eqref{eq:psi_wave}. Note that the PV is assumed to be a quantity localized on the top and bottom boundaries with associated phase-shifts and amplitude ratios relating to the CRWs, while the streamfunction has a vertical structure, and the linear superposition of $\hat{\psi}_{T,B}$ will lead to a structure that has a phase-tilt $\Delta\epsilon_{\rm eigen}$ and amplitude ratio $\Gamma_{\rm eigen}$ that does not in general coincide with that of the CRW phase-shift $\Delta\epsilon$ and amplitude ratio $\Gamma$ \cite{Mak-et-al24}. We see here for example that for the $\delta < 0$ (bottom row) there are more contours concentrated at the bottom boundary, and vice-versa for $\delta > 0$ (top row), consistent with our observations above that the finite-time transient growth is such that it wants to have a weaker bottom wave when $\delta < 0$ (and vice-versa for $\delta > 0$). At smaller $k$ (left column) the displayed streamfunctions have very small phase-tilts $\Delta\epsilon_{\rm eigen}$, somewhat in contrast to the corresponding phase-shifts $\Delta\epsilon$ where it is not that small (red squares in left column of Fig.~\ref{fig:fig5}). With increasing $k$ (left going to right columns) the phase-tilts (as well as the corresponding phase-shifts) seem to be increasing, in contrast to the normal modes (where they exist), which seem to want to \emph{decrease} in the phase-shift as well as the phase-tilt (Fig.~\ref{fig:fig4} here with the black circles, or Fig.~5 and 6$b$ of \cite{Mak-et-al24}). With increasing $k$, the streamfunctions become increasingly localized to the top and bottom boundaries, because of the behavior of the vertical structure function (e.g., Eq.~\ref{eq:psi_wave}) and shorter waves have a weaker interaction given the larger decay from the source. In the cases where there is normal mode instability (panels $a,b,d,e,f,i$), the streamfunction tilts against the shear, which is a feature that is seen and expected from the normal mode instabilities, while in the complementary case (panels $c,g,h$) the phase-tilts are close to zero or even tilting with the shear.

\begin{figure}
  \centering
  \includegraphics[width=1.0\linewidth]{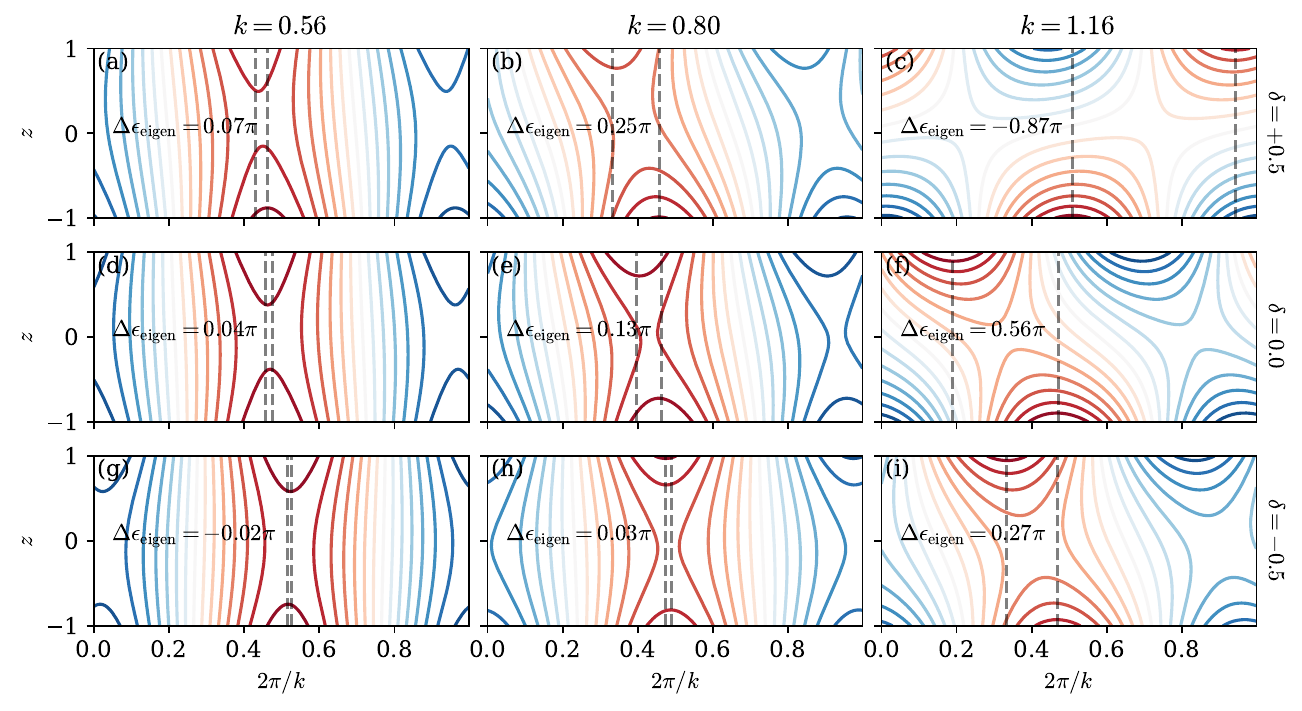}
  \caption{Streamfunction structure of the optimal perturbation at $t=5$ under the enstrophy norm for different $\delta$ (rows, $\delta=+0.5, 0.0, -0.5$ from top to bottom) and three choices of $k$ (columns, $k$ corresponding to the optimal normal mode growth of $\delta=+0.5, 0.0, -0.5$ from left to right). The corresponding phase tile $\Delta\epsilon_{\rm eigen}$ is indicated and is not equal to that of the CRW phase-shift $\Delta\epsilon$.}
  \label{fig:fig6}
\end{figure} 

While it is the optimal finite time streamfunctions for the enstrophy norm that is shown, the features described above qualitatively hold for the energy norm. The main observation is that the phase-tilts $\Delta\epsilon_{\rm eigen}$ in the energy norms seem to be larger or possibly equal in magnitude than the ones in the enstrophy norm for fixed $k$ and $\delta$, which is also an observation we have made for the phase-shifts $\Delta\epsilon$.



\section{Conclusion and outlooks}\label{sec:conclude}

In this work we investigated the characteristics of baroclinic instability with finite time transient growth in the sloped Eady problem and its representation in terms of Counter-propagating Rossby Waves (CRWs). The CRWs mechanism interprets instability as the constructive interference between a pair of CRWs, which depends crucially on the phase-shift $\Delta\epsilon$ and the magnitude ratios $\Gamma$. Our particular focus here are on the cases where the CRWs are asymmetric, with the bottom one controlled by a non-dimensional slope parameter $\delta$: if $\delta < 0$, the resulting slope is such that it \emph{strengthens} the bottom background potential vorticity gradient and leading to a \emph{stronger} bottom CRW for exponentially growing normal modes growth, vice-versa for $\delta > 0$, and reduces to the standard Eady problem with symmetric CRWs when $\delta=0$ \cite{Mak-et-al24}. For the finite time transient growth a choice of norm is required, and while there are in principle an infinity of possible choices, we have focused on quantities that are regarded to have general physical relevance --- namely the enstrophy norm, energy norms (total, kinetic, available potential) and boundary pressure norms. We only present the results for enstrophy and total energy norms since the results in those two norms are qualitatively representative of the others that we have considered. 

For the instantaneous growth (infinitesimal time limit), we find that the CRW quantities are independent of the asymmetry parameter $\delta$, and the CRWs have equal amplitudes (amplitude ratio of unity or $\Gamma = \pi/4$ here) if the norms are real and symmetric (the precise conditions are given in Proposition~\ref{prop:1}), in contrast to the results for the normal mode growths reported in \cite{Mak-et-al24} and reviewed in this work. This $\delta$-independence has its mathematical origins in the fact the operator associated with the relevant eigenvalue problem is symplectic in form even if it is $\delta$-dependent, leading to certain cancellations, and can be physically rationalized as the fact that, for instantaneous growth, the interaction and phase-locking aspect between the CRWs is less important than their initial configuration. Additionally, we find that the energy norm is somehow special in that the associated instantaneous growth rate remains well-defined in the $k\searrow 0$ limit for all choices of $\delta$ (the other norms considered here have unbounded growth rates in the same limit generically), and the instability characteristics is in fact completely $\delta$-independent (see Proposition~\ref{prop:2}). The latter observation relates to the governing operator for the eigenvalue problem having no $\delta$-dependence whatsoever. We find also that, numerically, the instantaneous growths are larger than the corresponding normal mode exponential growths (although the instantaneous growths are not sustained over time), but they only ever seem to agree at one particular value of the zonal wavenumber $k$ in the $\delta = 0$ case (with different values of $k$ depending on choice of norm), at least for the norms we have investigated in this work. The location where the two growth rates agree for a choice of norm relates to when the governing operator is normal (i.e., when the governing operator $\tilde{\mbox{A}}$ and its Hermitian conjugate $\tilde{\mbox{A}}^\dagger$ commute), and seems to indicate that for $\delta \neq 0$ the operator is generally \emph{non}-normal. This is likely a result that can be proven (e.g., showing that $\tilde{\mbox{A}}\tilde{\mbox{A}}^\dagger - \tilde{\mbox{A}}^\dagger \tilde{\mbox{A}}$ has a norm bounded above zero or that the components cannot be simultaneously zero when $\delta \neq 0$), although this was not something we have pursued beyond the trivial case (where $\mbox{M} = \mbox{I}$). We note that the value of $k$ where the operator is normal does not coincide with the value of $k$ of the optimal normal mode, at least for the choices of norm we have considered; it is presumably plausible to construct a norm such that there is coincidence (or show it cannot exist), but is also not something we have pursued in the most general case.

For finite time transient growth, we find that the effective growth rate (defined in Eq.~\ref{eq:sig_eff}) is bounded above by the instantaneous growth rate and below by the normal mode growth rate. Consistent with that observation, where the operator is normal for the choice of norm, all growths (instantaneous, finite time and exponential normal mode) are equal, but this only seem to occur for the $\delta = 0$ case in our numerical investigation (see Fig.~\ref{fig:fig4} here). The effective growth rates for finite time optimal growth for $\delta \neq 0$ have a `bouncing' behavior, where in the $k\searrow 0$ limit the growth rates oscillate between zero and some value (some constant in the enstrophy norm, some value decreasing to zero in the energy norm), with decreased spacing in $k$ between the `bounces'. Rather than it being a $\delta \neq 0$ behavior, we rationalize this as a phenomenon when there is no normal mode instability. There is no attractor in the system in such a regime, so the CRWs are essentially freely propagating, and depending on the choice $t$ and $k$ can lead to phases with constructive as well as destructive interference. The `bouncing' in the growth rate is then simply a compensation of those two phases and depends on the target time $t$ and $k$: for fixed $k$ and increasing $t$, the CRWs have more time to complete the cycles, while for fixed $t$ and \emph{decreasing} $k$ the CRWs travel faster (since these are Rossby waves with phase speed scaling like $k^{-1}$), so more cycles get completed for a given time. While not shown, the same bouncing effect is also present for $k\gg 1$ (where the `bounces' are progressively further apart in $k$ for the converse reason stated just now), and similar patterns can be seen in the CRW quantities (cf. `sawtooth' patterns in Fig.~6$b$ of \cite{HeifetzMethven05}). In terms of CRW amplitude ratios, for $\delta < 0$ the optimal initial conditions for finite time transient growth generically wants to have a \emph{weaker} bottom CRW, and vice-versa for $\delta > 0$, and is in the `opposite' sense to the normal mode CRW configurations. This was rationalized as a result of a larger gain in some choice fo norm that can result by growing a CRW over a location with stronger potential vorticity gradient, and vice-versa. A corollary observation is then the optimal initial conditions for finite time transient growth generally do \emph{not} coincide with the configuration where there is normal mode exponential growth in the $t\to\infty$ limit, and care is needed for the interpretation: when there is an attractor (in this case the normal mode configuration), it is certainly true the finite time optimal configuration will evolve \emph{towards} that in the long time limit, but that does not mean the configuration associated with the normal mode is optimal in the chosen norm for \emph{finite} target times (unless the operator happens to be normal).

In the present work we have shown results only for the sloped Eady problem, for the case where the meridional wavenumber $l=0$, and only for the enstrophy and total energy norms, partly for comparison with a previous that employs utilizes the same model for normal mode instability \cite{Mak-et-al24}. It was briefly reported that the instability characteristics for the available potential energy norm is qualitatively similar to that reported for the enstrophy norm, while those for kinetic energy and boundary pressure norm were qualitatively similar to the total energy norm; the reader can see additional results in the analysis code provided. All our qualitative findings seem to hold for the $l > 0$ case (which modifies the value of $\mu$), and a non-zero $l$ seems to only shift the normal mode instability regions (cf. \cite{Mak-et-al24}) with corresponding quantitative changes to the transient growth behavior. Further, all qualitative findings, such as the $\delta$-invariance in the CRW quantities for a certain class of norms, a seemingly special role of the total energy norm, and the CRW amplitude ratios for optimal finite time initial conditions all hold in an analogous sloped Phillips problem (two layer quasi-geostrophic system with a vertical shear; see Appendix~\ref{app:A} for the formulation). While our analysis has focused on the $\beta = 0$ case with an analogous $\delta \neq 0$ as that used in the sloped Eady problem here, a preliminary investigation finds the results also hold in the $\beta > 0$ but $\delta=0$ case (where the planetary $\beta$ introduces additional asymmetry in the CRW interaction by simultaneously making the westward/retrograde CRW propagate faster in the top layer as well as reduce the bottom eastward/prograde CRW's phase speed in the bottom layer).

The present work on transient instability focuses on a baroclinic system that can be completely described in terms of CRWs and that there are a low number of CRWs present in the system. Notable extensions could involve an analogous investigation into the optimal excitation of CRWs when the present sloped Eady problem does not have zero interior potential vorticity (e.g., when planetary $\beta$ is not zero), or cases with internal/continuous shear and thus potential vorticity reversals in the interior, which is often argued to be of more physical relevance in the ocean \cite[e.g.,][]{Isachsen11, Sterl-et-al24, Sterl-et-al25, Lobo-et-al25, PudigSmith25, DengWang26}. In those cases the CRWs are no longer small in number, nor are the locations of the dominant CRWs known \emph{a priori} (although the locations of large potential vorticity gradients and critical layers are likely candidates), there are still known procedures to investigate the problem analytically and numerically \cite[e.g.,][]{Heifetz-et-al04a, HeifetzMethven05}, and often low-order systems involving a smaller number of dominant CRWs can be derived that captures the essence of the problem \cite[e.g.,][]{Heifetz-et-al04b, DeVriesOopsteegh07a, DeVriesOopsteegh07b, DeVries-et-al09, Heifetz-et-al15}. Analogous investigations, extensions thereof and its interpretation in terms for CRWs can in principle be considered in that of stratified shear flows \cite[e.g.,][]{Rabinovich-et-al11, GuhaLawrence14} or in magnetohydrodynamic systems \cite[e.g.,][]{Heifetz-et-al15, Fraser-et-al26}, and it would be interesting to see what features reported here hold in the other systems with different physics.

The results found here for transient growth for the sloped Eady problem also holding in an analogous Phillips problem is intriguing, and one wonders whether this is related to something more fundamental common to those systems. One speculation we make here is that the observations may be related to the underlying parity-time ($\mbox{PT}$) symmetry of the system, which was shown for the Phillips problem in the work of \cite{David-et-al22}, and in the (sloped) Eady problem in \cite{Mak-et-al24}. As noted in \cite{Mak-et-al24}, there seems to be a close link between properties associated with $\mbox{PT}$ symmetry (Krein signatures, Krein collisions at exceptional points, spontaneous breaking of $\mbox{PT}$ symmetry and appearance modes with conjugate eigenvalues), linear instability (positive/negative energy modes, marginal stability points, appearance of normal mode instabilities with conjugate eigenvalues) and CRWs (positive/negative wave action, bifurcation points, phase-locking with constructive interference); some of these features can be described by a Hamiltonian formulation of the present problem and phrased in terms of wave actions, and we refer the reader to Appendix~\ref{app:B} for this. Given the stated link, it would if anything be surprising if the discrete symmetries have nothing to do with the reported properties relating to the instability characteristics. Our preliminary investigation based on the settings investigated in \cite{David-et-al22} have found that in the two cases where other discrete symmetries hold, such properties reported here do \emph{not} hold, although something else does. Related investigations in the Phillips setting and in the magnetohydrodynamic settings are underway and will be reported elsewhere; it would be interesting to see what features of the instability problem are general and common because of the mathematical structures present in the systems, despite the differences in the underlying physics.





%
%


%


\begin{acknowledgments}


This work was supported by the the RGC General Research Fund 16303625. The Python Jupyter notebook used to generate and analyze data for the sloped Eady problem may be found at \url{https://github.com/julianmak/julianmak.github.io/blob/master/files/Eady/Eady_transient_analysis.ipynb}. An analogous Python Jupyter notebook for a sloped Phillips problem may be found at \url{https://github.com/julianmak/julianmak.github.io/blob/master/files/2LQG/Phillips_transient_analysis.ipynb}.

Authorship is alphabetical after the first author. Resources, Supervision, Project administration, Funding acquisition: JM. Conceptualization: JM, EH. Methodology: JY, EH, JM. Visualization, Software, Formal analysis, Writing - Original Draft: JY, JM. Validation, Writing - Review \& Editing: everyone.


\end{acknowledgments}


\appendix

\section{Sloped Phillips problem}\label{app:A}

Based on the set up in \cite{David-et-al22} and \cite{ChenKamenkovich13} and using a notation analogous to that in the main body of this work, the dimensional linearized two-layer QG equations with equal layer depth on a $\beta$-plane (with $f = f_0 + \beta y$) over a uniformly sloping bottom are given by
\begin{linenomath*}
\begin{subequations}\begin{align} 
  \left( \frac{\partial}{\partial t} + U \frac{\partial}{\partial x} \right) q_T + \left( \beta + \frac{8}{L_d^2} U \right) \frac{\partial \psi_T}{\partial x} &= 0,\\ 
  \left( \frac{\partial}{\partial t} - U \frac{\partial}{\partial x} \right) q_B + \left( \beta + \frac{f_0}{D}\frac{\partial H_b}{\partial y}- \frac{8}{L_d^2} U \right) \frac{\partial \psi_B}{\partial x} &= 0,
\end{align}\end{subequations}
\end{linenomath*}
with the potential vorticity defined by
\begin{linenomath*}
\begin{subequations}\begin{align}
  q_T &= \nabla^2\psi_T +\beta y + \frac{4}{L_d^2}(\psi_B-\psi_T),\\
  q_B &= \nabla^2\psi_B +\beta y + \frac{4}{L_d^2}(\psi_T-\psi_B) + \frac{f_0}{D}H_b(x,y).
\end{align}\end{subequations}
\end{linenomath*}
The subscripts $T$ and $B$ denote the top and bottom layer variables, $\psi_{T,B}$ is the streamfunction, $D$ is the layer depth, and $H_b(x,y)$ represents the bottom topography. Scaling by the advective time-scale and considering the non-dimensional variables (with the tildes on top) $\tilde{t} = (U/L_d) t$, $\tilde{k} = k L_d$ and $\tilde{K}^2 = K^2 L_d^2 = (k^2 + l^2) L_d^2$, writing the resulting system in terms of $\hat{q}_{T,B}$ instead of the usual $\psi_{T,B}$ (e.g., \cite{David-et-al22}) leads to the non-dimensional system (dropping all tildes for convenience)
\begin{linenomath*}
\begin{subequations}\begin{align}
  \left( \frac{\partial}{\partial t} + \mathrm{i}k \right) \hat{q}_T - \frac{4\mathrm{i}k(\gamma + 2)}{K^2(K^2 + 8)} \left[ \left(4 + K^2\right) \hat{q}_T + 4\hat{q}_B \right] &= 0,\\
  \left( \frac{\partial}{\partial t} - \mathrm{i}k \right) \hat{q}_B - \frac{4\mathrm{i}k(\gamma + \delta - 2)}{K^2(K^2 + 8)} \left[ 4\hat{q}_T + \left(4 + K^2\right) \hat{q}_B \right] &= 0,
\end{align}\end{subequations}
\end{linenomath*}
where $\delta = (f_0 L_d^2/(4UD))\partial H_b / \partial y$ is a non-dimensional parameter relating to the bottom slope, and $\gamma = \beta L_d^2 / (4U)$ is a criticality parameter relating to the flow shear needed to overcome the stabilizing effect of $\beta$ . The governing operator $\mbox{A}$ analogous to that of \eqref{eq:A_equ} would be
\begin{linenomath*}
  \begin{equation}
  \mathrm{A} = \frac{-\mathrm{i}k}{K^2(K^2 + 8)}
  \begin{pmatrix}
    K^4-4\gamma(K^2+4)-32 & -32-16\gamma \\
    32-16(\gamma+\delta) & -K^4-4(\gamma+\delta)(K^2+4)+32
  \end{pmatrix}.
  \end{equation}
\end{linenomath*}
Note again the governing system is $\mbox{PT}$ symmetric (as was shown and investigated by \cite{David-et-al22}) in the sense highlighted here (one minus sign from $k$ that compensates with a minus sign from the $\mathrm{d}/\mathrm{d}t$, with none from $K^2$), even in the presence of the $\gamma$ and $\delta$ parameter that gives rise to asymmetry.


\section{Wave activity invariants and a Hamiltonian formulation}\label{app:B}

The pseudomomentum and pseudoenergy and its partition in terms of CRWs associated with the Eady problem seems to have been partially or implicitly implied in previous works but not stated explicitly, and we take this opportunity to do so here. Using the bracket notation to again denote an integral (with appropriate complex conjugates, following some of the notation of \cite{Heifetz-et-al04a} but taking a different sign convention (counter to \cite{Heifetz-et-al04a, Held85}), the pseudomomentum $\mathcal{P}$ and pseudoenergy $\mathcal{H}$ can be shown to be (noting that we include a factor of $1/2$ in our definition of the bracket)
\begin{linenomath*}
\begin{align}
  \mathcal{P} = -\left\langle \eta, \frac{\partial \overline{Q}}{\partial y}\eta \right\rangle, \qquad \mathcal{H} = E + \left\langle \eta,  (\overline{U} - \overline{U}_m)\frac{\partial \overline{Q}}{\partial y} \eta \right\rangle,
\end{align}
\end{linenomath*}
where $E$ is the total energy in terms of the perturbation quantities, and $\overline{U}_m$ is some reference value of the base state velocity. The key variable is $\eta$ the meridional displacement associated with $v = \partial\psi / \partial x$, and has the relation $q = \eta (\partial \overline{Q} / \partial y)$.

In terms of CRWs, starting first with $\mathcal{P}$, we define
\begin{linenomath*}
\begin{align}
  \mathcal{P}_{ij} = -\left\langle \eta_i, \frac{\partial \overline{Q}}{\partial y} \eta_j \right\rangle = -\left\langle \frac{q_i}{\partial \overline{Q} / \partial y}, q_j \right\rangle,
\end{align}
\end{linenomath*}
with $\{i, j\} = \{T ,B\}$ depending on whether it is the top or bottom CRW we are talking about. In the general case for the CRWs $q$ is modeled to be delta-functions, so $\mathcal{P}_{ij} = 0$ by construction if $i\neq j$, which is orthogonality in the pseudomomentum inner product \cite[cf.,][]{Held85}. In the Eady system $\partial\overline{Q} / \partial y = 0$ everywhere except at the top and bottom boundaries, and in our case we simply have (in the non-dimensionalized form) $\partial\overline{Q} / \partial y = 1, -(1-\delta)$ respectively, so $\mathcal{P}_{TT} = -T^2/2$ and $\mathcal{P}_{BB} = +B^2/[2(1-\delta)]$. Total $\mathcal{P}$ conservation in this case implies that (cf. Eq.~25 of \cite{Heifetz-et-al04a})
\begin{linenomath*}
\begin{align}\label{eq:pseudo-mo}
  \mathcal{P} = \mathcal{P}_T + \mathcal{P}_B = -\frac{T^2}{2} + \frac{B^2}{2(1-\delta)} = \textnormal{constant},
\end{align}
\end{linenomath*}
which is the first constant of motion in \eqref{eq:const_of_motion}. The first and second term here can be regarded as the pseudomomentum of the top and bottom CRW respectively, and there is a clean separation because of the imposed orthogonality. The sign convention taken here means that the top west/leftward CRW (propagating counter to the east/rightward mean flow) has \emph{negative} pseudomomentum, and vice-versa for the bottom wave.

Proceeding to $\mathcal{H}$, we define the brackets
\begin{linenomath*}
\begin{align}
  \mathcal{S}_{ij} = -\left\langle \eta_i, (\overline{U} - \overline{U}_M)\frac{\partial \overline{Q}}{\partial y} \eta_j \right\rangle = -\left\langle \frac{q_i}{\partial \overline{Q} / \partial y}, (\overline{U} - \overline{U}_M) q_j \right\rangle,
\end{align}
\end{linenomath*}
where $\overline{U}_m$ is an appropriate reference zonal velocity, and $S_{ij} = 0$ for $i\neq j$. In the present Eady problem, because our choice of symmetrizing the problem and non-dimensionalization, we have $\overline{U} = \pm1$ at the top and bottom, and we can take $\overline{U}_m = 0$, and so $\mathcal{S}_{TT} = -T^2/2$ while $\mathcal{S}_{BB} = -B^2 / [2(1-\delta)]$ (note that resulting signs). The expression for $E$ in terms of CRWs come from the expression for $\mathbf{q}^\dagger \mbox{M} \mathbf{q}$ with $\mbox{M}$ from \eqref{eq:ene_norm}, which is
\begin{linenomath*}
\begin{align*}
  E = \frac{1}{2}\frac{(T^2 + B^2)\cosh 2\mu + 2TB \cos\Delta\epsilon}{\mu \sinh 2\mu},
\end{align*}
\end{linenomath*}
and so the pseudoenergy is
\begin{linenomath*}
\begin{align}\label{eq:pseudo-ene}
  \mathcal{H} = \frac{1}{2}\left[\frac{(T^2 + B^2)\cosh 2\mu + 2TB \cos\Delta\epsilon}{\mu \sinh 2\mu} - \left(T^2 + \frac{B^2}{1-\delta}\right)\right] = \textnormal{const},
\end{align}
\end{linenomath*}
which gives the second constant of motion in \eqref{eq:const_of_motion}.

Note that these quantities are generally sign-indefinite, but under conditions where they are sign-definite these correspond to the Arnol'd type stability conditions: this occurs in the sloped Eady probably for example when $\delta > 1$ and $\mathcal{P} < 0$, which corresponds to the case where both CRWs are leftward traveling, and there is an analogous condition for $\mathcal{H}$ to be sign-definite that depends on the wavenumber $\mu$. While we have diagnosed how the pseudomomentum and pseudoenergy evolve in instability growth for the present Eady problem, we have not found any particular feature that is particularly illuminating for the present problem. One possible reason for this is perhaps the problem is sufficiently simple such that the constraints ends up implying that there is nothing too complex that can occur. In cases beyond the Eady problem the corresponding forms are presumably not so simple and the evolution of the wave activities may be prove to be fruitful (e.g., the Charney problem considered in \cite{Heifetz-et-al04b}, possibly for the MHD problem considered in \cite{Heifetz-et-al15}).

We note that if we define the action as
\begin{linenomath*}
\begin{align}
  J_T = -\mathcal{P}_T = \frac{T^2}{2}, \qquad J_B = +\mathcal{P}_B = \frac{B^2}{2(1-\delta)}
\end{align}
\end{linenomath*}
and take the angle to be respective phases, then with the re-scaled pseudoenergy \eqref{eq:pseudo-ene} as the Hamiltonian given by
\begin{linenomath*}
\begin{align}
  H = \frac{k\mathcal{H}}{2} = \frac{k}{2}\left[\left(\frac{\cosh 2\mu}{\mu \sinh 2\mu} - 1\right)J_T + \left(\frac{(1-\delta)\cosh 2\mu}{\mu \sinh 2\mu} - 1\right)J_B + \frac{2\sqrt{1-\delta}}{\mu\sinh 2\mu} \sqrt{J_T J_B} \cos\Delta\epsilon\right],
\end{align}
\end{linenomath*}
the amplitude-phase equations in \eqref{eq:amp_phase} are equivalent to
\begin{linenomath*}
\begin{align}
  \frac{\partial \epsilon_T}{\partial t} = +\frac{\partial H}{\partial J_T}, \quad
  \frac{\partial J_T}{\partial t} = -\frac{\partial H}{\partial \epsilon_T}, \qquad
  \frac{\partial \epsilon_B}{\partial t} = -\frac{\partial H}{\partial J_B}, \quad
  \frac{\partial J_B}{\partial t} = +\frac{\partial H}{\partial \epsilon_B}.
\end{align}
\end{linenomath*}
Here the conservation of $\mathcal{P} = -J_T + J_B$ corresponds to the continuous symmetry $\epsilon_{T,B} \mapsto \epsilon_{T,B} + \theta$. The corresponding Poisson bracket here is diagonal, but the Hamiltonian has an interaction term $\sqrt{J_T J_B} \cos\Delta\epsilon$.

It can be shown there is a set of symmetrized variable where a corresponding Hamiltonian is diagonal instead, and can be related to the positive/negative energy waves description \cite[e.g.,][]{Cairns79}, but the analogous Poisson bracket there generically has an interaction term. This fact relates to the Hamiltonian being sign-indefinite and having a hyperbolic normal form, and a unstable Krein pair cannot be reduced to independent oscillators but only to conjugate pairs, where either the Hamiltonian or the Possion bracket is diagonal. The Hamiltonian picture is perhaps more natural for studying the consequences of $\mbox{PT}$ symmetry on the system in the present and analogous systems governed by non-Hermitian operators, and is something that we will explore in the future.


\section{Other choices of norms}\label{app:C}

The following norms based on physical quantities (motivated by those studied in \cite{KalashnikChkhetiani18}) were also considered in the analysis as part of the present investigation, although the results were not presented in the main article given the qualitative similarity with the enstrophy and energy norms. Using the bracket notation to again denote an integral (with appropriate complex conjugates):
\begin{itemize}
  \item Kinetic energy norm (with solutions qualitatively like the energy norm),
  \begin{linenomath*}
  \begin{align}
    \mbox{M} = \begin{pmatrix}
    \cfrac{\cosh 2\mu}{2\mu\sinh 2\mu} + \cfrac{1}{\sinh^2 2\mu} &
    \cfrac{\cosh 2\mu}{\sinh^2 2\mu} + \cfrac{1}{2\mu\sinh 2\mu} \\
    \cfrac{\cosh 2\mu}{\sinh^2 2\mu} + \cfrac{1}{2\mu\sinh 2\mu} &
    \cfrac{\cosh 2\mu}{2\mu\sinh 2\mu} + \cfrac{1}{\sinh^2 2\mu}
    \end{pmatrix} \qquad \Leftrightarrow \qquad K = \left\langle \nabla \psi, \nabla \psi \right\rangle.
  \end{align}
  \end{linenomath*}
  
  \item Available potential energy norm (with solutions qualitatively like the enstrophy norm),
  \begin{linenomath*}
  \begin{align}
    \mbox{M} = \begin{pmatrix}
    \cfrac{\cosh 2\mu}{2\mu\sinh 2\mu} - \cfrac{1}{\sinh^2 2\mu} &
    -\cfrac{\cosh 2\mu}{\sinh^2 2\mu} + \cfrac{1}{2\mu\sinh 2\mu} \\
    -\cfrac{\cosh 2\mu}{\sinh^2 2\mu} + \cfrac{1}{2\mu\sinh 2\mu} &
    \cfrac{\cosh 2\mu}{2\mu\sinh 2\mu} - \cfrac{1}{\sinh^2 2\mu}
    \end{pmatrix} \qquad \Leftrightarrow \qquad P = \left\langle \frac{\partial \psi}{\partial z}, \frac{\partial \psi}{\partial z} \right\rangle,
  \end{align}
  \end{linenomath*}
  and is the difference between the $\mbox{M}$ for total energy given in \eqref{eq:ene_norm} and the kinetic energy $\mbox{M}$ given above.
  
  \item Boundary pressure (with solutions qualitatively like the energy and kinetic energy norm),
  \begin{linenomath*}
  \begin{align}
    \mbox{M} = \frac{1}{(\mu\sinh 2\mu)^2}
      \begin{pmatrix}
      \cosh^2 2\mu + 1 & 2\cosh 2\mu \\
      2\cosh 2\mu & \cosh^2 2\mu + 1
      \end{pmatrix} \qquad \Leftrightarrow \qquad \hat{P} = \left\langle \psi(z=+1), \psi(z=+1) \right\rangle + \left\langle \psi(z=-1), \psi(z=-1) \right\rangle.
  \end{align}
  \end{linenomath*}
\end{itemize}

All of the above norms satisfy the conditions of Proposition~\ref{prop:1} and are $\delta$-invariant in the stated fashion.





%

\end{document}